\documentclass[acmlarge,nonacm]{acmart}

\hypersetup{colorlinks=true, linkcolor=black, citecolor=black, urlcolor=blue}

\usepackage{xcolor}
\usepackage{graphicx}
\usepackage{booktabs}
\usepackage{multirow}
\usepackage{mathtools}
\usepackage{mathrsfs}
\usepackage{stmaryrd}
\usepackage{physics}
\usepackage{bbm}
\usepackage{dsfont}
\usepackage{algorithm}
\usepackage{algorithmicx}
\usepackage{algpseudocode}
\usepackage{listings}
\usepackage{xparse}
\usepackage{tcolorbox}
\usepackage{wrapfig}
\usepackage{tikz}
\usetikzlibrary{positioning}
\usepackage{thmtools}
\usepackage{thm-restate}
\usepackage{anyfontsize}
\usepackage{quantikz}

\newif\ifdiff
\difffalse    

\newcommand{\diff}[1]{
  \ifdiff
    {\color{red}#1}
  \else
    #1
  \fi
}

\theoremstyle{acmplain}
\newtheorem{theorem}{Theorem}
\newtheorem{lemma}[theorem]{Lemma}

\newtheorem{definition}[theorem]{Definition}
\newtheorem{remark}[theorem]{Remark}

\newcommand{\set}[1]{\left\{#1\right\}}
\newcommand{\bits}{\set{0,1}}
\renewcommand{\tr}[1]{\mathrm{Tr}\!\left[ #1 \right]}
\newcommand{\ptr}[2]{\mathrm{Tr}_{#1}\!\left[ #2 \right]}
\newcommand{\pr}[1]{\mathrm{Pr}\!\left[ #1 \right]}

\renewcommand{\proj}[1]{\ketbra{#1}{#1}}

\newcommand{\bO}{\mathrm{O}}
\newcommand{\bbN}{\mathbb{N}}
\newcommand{\1}{\mathbbm{1}}
\newcommand{\ext}{\mathrm{ext}}
\newcommand{\Ext}{\mathrm{Ext}}

\newcommand{\cB}{\mathcal{B}}
\newcommand{\cC}{\mathcal{C}}
\newcommand{\cD}{\mathcal{D}}
\newcommand{\cE}{\mathcal{E}}
\newcommand{\cF}{\mathcal{F}}

\newcommand{\cH}{\mathcal{H}}

\newcommand{\cP}{\mathcal{P}}

\newcommand{\cS}{\mathcal{S}}

\newcommand{\hunps}[1][\default]{\mathrm{H}_{\mathrm{unp}\left({#1}\right)}}

\newcommand{\hmaxcompc}[1][\default]{\overset{c}{\mathrm{H}}_{\max}^{#1}}

\newcommand{\hmincompc}[1][\default]{\overset{c}{\mathrm{H}}_{\min}^{#1}}

\newcommand{\bP}{\Delta_P}

\newcommand{\dist}{\mathrm{d}}

\newcommand{\hill}{\textrm{HILL}}

\newcommand{\IP}{\mathrm{IP}}

\newcommand{\Prg}{\mathrm{PRG}}
\newcommand{\prg}{\mathrm{prg}}

\usepackage{xspace}

\usepackage[capitalize,noabbrev]{cleveref}

\usepackage{quantikz}

\begin{document}

\title{Fully Quantum Computational Entropies}

\author{Noam Avidan}
\authornote{These authors contributed equally to this work.}
\affiliation{
  \institution{The Center for Quantum Science and Technology, Faculty of Mathematics and Computer Science,  
  Weizmann Institute of Science}
  \city{Rehovot}
  \country{Israel}}
\email{Noam.cohen-avidan@weizmann.ac.il}

\author{Thomas A. Hahn}
\authornotemark[1]
\affiliation{
  \institution{The Center for Quantum Science and Technology,
  Faculty of Physics of Complex Systems,  
  Weizmann Institute of Science}
  \city{Rehovot}
  \country{Israel}}

\author{Joseph M. Renes}
\affiliation{
  \institution{Institute for Theoretical Physics, ETH}
  \city{Zurich}
  \country{Switzerland}}

\author{Rotem Arnon}
\affiliation{
  \institution{The Center for Quantum Science and Technology,
  Faculty of Physics of Complex Systems,  
  Weizmann Institute of Science}
  \city{Rehovot}
  \country{Israel}}

\begin{abstract}

    Quantum information theory, particularly its entropic formulations, has made remarkable strides in characterizing quantum systems and tasks. However, a critical dimension remains underexplored: computational efficiency.
    While classical computational entropies integrate complexity and feasibility into information measures, analogous concepts have yet to be rigorously developed in the quantum setting.
    In this work, we lay the basis for a new \emph{quantum computational information theory}.
    Such a theory will allow studying efficient -- thus relevant in practice -- manipulation of quantum information.
    We introduce two entropies: quantum computational min- and max-entropies (along with their smooth variants).
    Our quantum computational min-entropy is both the fully quantum counterpart of the classical unpredictability entropy, as well as the computational parallel to the quantum min-entropy.  We establish a series of essential properties for this new entropy, including data processing and chain rules.
    The quantum computational max-entropy is defined via a duality relation and gains operational meaning through an alternative formulation that we
    derive. Notably, it captures the efficiency of entanglement distillation with the environment, restricted to local quantum circuits of bounded size.
    With the introduction of our computational entropies and their study, this work marks a critical step toward a quantum information theory that incorporates computational elements.
\end{abstract}

\maketitle
\pagestyle{plain}

\newpage

\section{Introduction and Motivation}
Quantum information theory (QIT) provides the mathematical setup that describes how information, encoded in physical quantum systems, can be stored, manipulated, and retrieved~\cite{Tomamichel_2016,khatri2024principlesquantumcommunicationtheory}.
As such, it serves as the underlying theory upon which studies in quantum computation, communication, and cryptography are performed and has a crucial impact on the development of quantum technologies~\cite{Tomamichel_2016, khatri2024principlesquantumcommunicationtheory,Dupuis_2014_one_shot_decup,DFRR20_EAT,renner2006securityquantumkeydistribution,Horodecki_2009_entanglement_measures, Kimble_2008}.

\subsection*{The need for a new theory: computational quantum information theory}
In this work, we advocate towards the need to extend QIT in order to include computational aspects. So far, the study of QIT has focused on understanding the nature of quantum information without deeply considering computational limitations; see, e.g.~\cite{Tomamichel_2016,khatri2024principlesquantumcommunicationtheory}.
Alas, in the real world, any manipulation of quantum information is and will be done with some constrained computational abilities. For example, information will be manipulated using a quantum computer that can run for only a fixed, finite amount of time. Studying efficient -- thus relevant in practice -- manipulation of quantum information is therefore of great importance. We term the new theory we advocate for \emph{computational QIT}.
We place in this work a major cornerstone in the pursuit of this novel theory, as we elaborate later on.
For now, let us first start by examining how different a computational QIT can be from the standard, well-known QIT.

A strikingly simple example is the Schrödinger-HJW theorem~\cite{Schr_pur_equive_1936,HJW_purificatioin_equive_199314} or purification theorem: Consider a mixed quantum state $\rho_A$. There are various different \emph{pure} states $\psi_{AB}$ that purify $\rho_A$ (i.e., states for which $\ptr{B}{\psi_{AB}}=\rho_A$). The purification theorem states that all purifications $\psi_{AB}$ are equivalent up to local isometries on the purifying register $U_B$ on $\psi_{AB}$.
This equivalence, mostly taken for granted,
has many implications. For example, applying local isometries cannot increase the amount of entanglement between $A$ and $B$, or the information (measured by entropies) that one system has on the other~\cite{plenio2006introductionentanglementmeasures, Horodecki_2009_entanglement_measures}.

What happens, however, if we now allow ourselves to consider only computationally efficient operations?
Informally, when we say ``efficient operations'', we refer to all the operations that can be implemented in practice in a reasonable amount of time. One can formalize this in the language of quantum computation~\cite{nielsen2010quantumbook} (as we do later) or alternatively consider other notions of computational power, e.g., by allowing only noisy gates (NISQ era~\cite{preskill2018quantum}) or experimentally feasible operations.

To go back to our question-- the situation is quite complex, and the purification theorem does not seem to hold as is. It is believed that switching from one purification to another by finding the local transformation~$U_B$ is computationally hard in general. For example, the computational complexity of the Uhlmann transformation problem~\cite{uhlmann1976transition} has been investigated in~\cite{bostanci2023unitary}, with strong indication that approximating the unitary transformation that maximizes the fidelity between purifications of given states is computationally hard. 

The results presented in~\cite{bostanci2023unitary}, alongside works studying concepts such as pseudo-entanglement~\cite{gheorghiu2020estimating,aaronson2023quantumpseudoentanglement,arnonfriedman2023computationalentanglementtheory}, have far-reaching consequences for QIT. Basic mathematical tools and concepts like the purified distance~\cite{Tomamichel_2010_Duality}, entanglement measures~\cite{Horodecki_2009_entanglement_measures}, the operational meaning of entropies~\cite{KRS09OperationalMeaningEntropy}, etc. -- all need to be revisited to include computational aspects.

Since computational power is of relevance for any future real-world quantum technology, we are obligated to consider the ways in which quantum information can be manipulated efficiently. A \emph{computational QIT} must be developed.
This new theory will then allow us to consider tasks such as efficient entanglement distillation~\cite{arnonfriedman2023computationalentanglementtheory}, quantum cryptography with a computationally bounded adversary~\cite{alagic2016computational}, and more.
Additionally, such a theory should set the mathematical foundations and encompass topics that are of growing interest, such as pseudo-random quantum states~\cite{QPRSJKF18}, pseudo-random unitaries~\cite{bouland2023publickeypseudoentanglementhardnesslearning,FermiMaHsin24pseRandom_Unitaries,schuster2024random}, and computational entanglement~\cite{gheorghiu2020estimating,arnonfriedman2023computationalentanglementtheory, aaronson2023quantumpseudoentanglement,leone2025entanglementtheorylimitedcomputational}.

\subsection*{The need for new entropies: computational entropies}

Quantum entropies are the bread and butter of QIT~\cite{KRS09OperationalMeaningEntropy,RenerWolf04smoothreny}.
Substantial effort is continuously being directed towards broadening and improving the toolkit of mathematical properties and relations that hold for the various quantum entropies~\cite{DFRR20_EAT,MDSFT13_quantum_renyi_entropies_dualiteis,MFSR22_GEAT,Dupuis_2015_chain_rules_smooth,Tomamichel_2010_Duality,arnon2018practical,arqand2024generalizedrenyientropyaccumulation,arqand2025marginalconstrainedentropyaccumulationtheorem,beigi2013sandwiched_dualiteis,zhai2023boundssmoothminmaxentropy,MerkulovArnon25_EATPQC} and the study of quantum entropies is extensive and rapidly evolving.
Nevertheless, to date, the literature has largely left untouched an important research avenue, namely \emph{quantum computational entropies}.
In our goal of forming a computational QIT, we put the emphasis in this work on defining quantum computational entropies.

Computational entropies~\cite{barak2003computational,HLR07SepPseudoentropyfromCompressibility,arnon2025computational,CC17Computationalminentropy} quantify information in the same sense that the information-theoretic entropies do, but they incorporate a computational dimension, such as efficiency or complexity.
As a simple example, one can take the min-entropy. The min-entropy, $\mathrm{H}_{\min}(A)$, of a  
random variable $A$, describes the (in)ability to predict the value of $A$.
Here, when we write ``predict'', we implicitly mean that any strategy that allows one to guess the value of $A$ is valid, \emph{not} taking into account the computational complexity of the guessing strategy.
In a computational variant of the min-entropy, on the other hand, we may restrict ourselves only to, e.g., efficient strategies.\footnote{Again, to formally discuss ``efficiency'', one needs to first have some notion of computational power and complexity in mind. For example, later on, we will consider quantum circuits with a given number of gates from a fixed set of universal gates.}

In the domain of theoretical computer science, the examination of computational entropies has been pursued for several decades now, see e.g.~\cite{barak2003computational,HLR07SepPseudoentropyfromCompressibility,KPWW14CounterexampleHILLchainrule,fuller2012computational,reingold2008dense}. Beyond their inherent significance, the computational entropies underpin studies of pseudo-randomness, randomness extraction, and numerous cryptographic applications~\cite{DP08LeakageResilientStandard,HILL1999pseudorandom,S16_betterHILL_chain_rule,fuller2012computational}.
As for quantum computational entropies, a pioneering work~\cite{CC17Computationalminentropy} proposed various definitions for quantum computational min-entropy.\footnote{For both classical and quantum information, there are multiple directions in which one can extend the min-entropy to the computational case. Different definitions have different interpretations and, a priori, all may be of relevance.}
The work investigated the extent to which classical computational entropy properties and statements could be extended to these newly suggested quantum variants.
The introduced quantities, regrettably, failed to satisfy some critical expectations, most notably an entropic chain rule.
This raised questions regarding the suitability of the proposed definitions for quantum information.

Another related work is that of~\cite{Munson_2025,Yunger_Halpern_2022}, in which computational complexity is considered for quantum thermodynamics. For that purpose, the authors examined the task of computationally limited hypothesis testing (for quantum states). To go beyond the task of hypothesis testing, some variant of a computational relative entropy was defined. 

This entropy can be used to quantify the optimal computational data-compression rate, but does not provide an explicit method with which to (even approximately) achieve this rate. 
Similarly, the relevance of this entropy to quantum decoupling\footnote{A specific form of decoupling, where certain subsystems are traced out, is considered. This type of decoupling cannot be used to study entanglement distillation~\cite{Dupuis_2014_one_shot_decup,khatri2024principlesquantumcommunicationtheory}.} was argued on the basis of a conjectured chain rule. The conjectured chain rule was then connected to the assumption that ``efficiently preparable pseudomixed states'' do not exist (see~\cite{Munson_2025} for the formal statement), without giving arguments for why this assumption should hold. 

In a recent work~\cite{arnon2025computational}, two of the authors of the current paper introduced a new computational variant of the min-entropy, which was coined the ``quantum computational unpredictability entropy''. The entropy was defined for classical-quantum (cq-) states and it extends both the quantum min-entropy (for cq-states) as well as the classical unpredictability entropy~\cite{HLR07SepPseudoentropyfromCompressibility}, which is a computational entropy.
The quantum computational unpredictability entropy is distinct from the entropies presented in~\cite{CC17Computationalminentropy} and~\cite{Munson_2025}, and it satisfies both the fully quantum leakage chain rule that was posed as an open question in~\cite{CC17Computationalminentropy} as well as the chain rule conjectured in~\cite{Munson_2025}. 
The fact that the newly defined entropy ``behaves as anticipated'' from the point of view of quantum information theory indicates that it may serve as a good starting point for studying computational quantum entropies more generally.

Here, we significantly broaden the study of quantum computational entropies.
Firstly, we extend the definition of unpredictability entropy~\cite{arnon2025computational} to general and fully quantum states; this leads to what we appropriately term \emph{the computational min-entropy}. Then, motivated by the duality of entropies~\cite{Tomamichel_2010_Duality}, we define another new quantity --  \emph{the computational max-entropy}.
We also provide the entropies' smooth counterparts -- the quantum computational smooth min- and max-entropies.

The full definitions are given in Section~\ref{sec:comp_min_max_entropies}.
For now,  let us remark that the definitions are very delicate and carefully considered -- intriguing issues arise when examining computational entropies, and nothing should be taken for granted. As mentioned before, a core example is that not all purifications of a state should be treated equally when computational power is limited, as the parties cannot simply apply generic unitaries on their states. In other words, Uhlmann's theorem~\cite{uhlmann1976transition} -- 
a key tool for studying quantum entropies and their properties -- can no longer be used as is.
Moreover, it is no longer clear that the computational min- and max- entropies can be related to the entropies based on computational hypothesis testing~\cite{Munson_2025}. This is in contrast to the information theoretic case~\cite{dupuis2014generalized}, again due to the core computational aspect.

We further expand the novel toolkit of quantum computational entropies by proving essential properties, such as (efficient) data processing inequalities and leakage chain rules for the computational min-entropy, and give an operational meaning to the computational max-entropy (different from the one for the information-theoretic max-entropy; see below).

Overall, this work presents a major cornerstone to the development of a computational QIT by examining its most basic entropic quantities. This adds to our previous related works~\cite{arnonfriedman2023computationalentanglementtheory,arnon2025computational}.
Of course, our research reveals only the tip of the iceberg; we supply interesting open questions to continue developing a computational QIT in future works.

\section{Main Results and Technical Overview}
In this work, we initiate a systematic study of new computational entropies in the fully quantum setting. 
Our central objects are the quantum computational min- and max-entropy, which extend both classical computational entropies and the information-theoretic quantum entropies. Below, we provide a high-level map of the main conceptual steps and results.
\subsection*{Computational Min-Entropy}
We begin with the computational min-entropy. In the (fully quantum) information-theoretic setting, where all registers hold quantum information, the conditional min-entropy $\mathrm{H}_{\min}(A|B)_\rho$ can be expressed in terms of the maximum achievable overlap with a maximally entangled state (the singlet fraction) after applying arbitrary quantum operations on $B$, see e.g.~\cite{KRS09OperationalMeaningEntropy}.
\begin{definition}[Min-Entropy~\cite{KRS09OperationalMeaningEntropy}]\label{def:intro:min-entropy-math}
    Let $\rho_{AB} \in S_{\bullet}(AB)$ be a sub-normalized bipartite quantum state.\footnote{Sub-normalized quantum states are positive semi-definite operators that satisfy the trace inequality $\tr{\rho} \leq 1$.} The conditional min-entropy of $A$ given $B$ is defined as:
\begin{align*}
    \mathrm{H}_{\min}(A|B)_\rho \coloneq  -\log d_{A} \max_{\cE_{B\to A'} }  F\left( (\mathbb{I}_A \otimes \cE)(\rho_{AB}), \proj{\Phi_{AA'}} \right) \; ,
\end{align*}
where $\cE_{B\to A'}$ are quantum channels from $B$ to some other Hilbert space $A^\prime$ and $\proj{\Phi_{AA'}}$ is the (canonical) maximally entangled state on $\cH_{A} \otimes \cH_{A^\prime}$.
\end{definition}

This operational viewpoint naturally suggests how to incorporate computational limitations. Instead of optimizing over all quantum channels on $B$, we restrict attention to those that can be implemented by circuits of bounded size (See also~\cref{fig:intro:Es-circuit-min-entropy}). This leads to our definition of the computational min-entropy $\hmincompc[s](A|B)_\rho$, along with its smooth variant.

\begin{definition}[Quantum Computational Min-Entropy]\label{def:intro:nosmooth_quantum_comp_entropy}
    For any sub-normalized bipartite quantum state ${\rho_{AB}\in \cS_{\bullet}(AB)}$, and $s \in \mathbb{N}$, we say that
    \begin{equation*}
        \hmincompc[s](A|B)_{\rho} \coloneq -\log d_{A} \max_{\cE^s_{B\to A'} }  F\left( (\mathbb{I}_A \otimes \cE^s)(\rho_{AB}), \proj{\Phi_{AA'}} \right) \;,
    \end{equation*}
    where $\proj{\Phi_{AA'}}$ is the (canonical) maximally entangled state on $\cH_{A} \otimes \cH_{A^\prime}$ and the optimization is over the set of all quantum channels from $B$ to $A'$ that can be implemented by circuits of size at most $s$. 
\end{definition}

\begin{figure}[ht]
    \centering
    \begin{quantikz}[row sep=0.3cm]
        \lstick[wires=8]{$\rho_{AB}$} & \qw & \qw                                      & \qw \rstick[wires=4]{$A$}  \\
        & \qw & \qw                                      & \qw                        \\
        & \qw & \qw                                      & \qw                        \\
        & \qw & \qw                                      & \qw                        \\
        & \qw & \gate[wires=4]{\mathcal{E}^s_{B \to A'}} & \qw \rstick[wires=4]{$A'$} \\
        & \qw & \qw                                      & \qw                        \\
        & \qw & \qw                                      & \qw                        \\
        & \qw & \qw                                      & \qw
    \end{quantikz}
    \quad $\approx$ \quad $\proj{\Phi_{AA'}}$
    \caption{A circuit diagram for the operational meaning of the computational min entropy. A channel $\mathcal{E}^s$ (circuit of size~$s$) acting only on the $B$ register. The goal of the circuit is to maximize the fidelity of the resulting global state with the maximally entangled state $F(\rho_{AA'},\proj{\Phi_{AA'}})$.}
    \Description[Computational min entropy diagram]{A circuit diagram for the operational meaning of the computational min entropy. A channel $\mathcal{E}^s$ (circuit of size~$s$) acting only on the $B$ register. The goal of the circuit is to maximize the fidelity of the resulting global state with the maximally entangled state $F(\rho_{AA'},\proj{\Phi_{AA'}})$.}
    \label{fig:intro:Es-circuit-min-entropy}
\end{figure}

For cq-states, we show that this definition is equivalent to the \textit{quantum computational unpredictability entropy}, which was introduced in a prior work~\cite{arnon2025computational}. Notably, for cq-states, the min-entropy quantifies an adversary's optimal guessing strategy of the output stored in $A$, given that they hold $B$~\cite{KRS09OperationalMeaningEntropy}. The computational analog, where one instead optimizes over all efficient guessing strategies, is captured by the quantum computational unpredictability entropy~\cite{arnon2025computational}, and therefore also by our more general definition.

We highlight the following features of this definition.
\begin{enumerate}
    \item This definition holds for any initial state; we do not need to impose any conditions on the computational power needed to generate~$\rho_{AB}$. We only impose computational constraints at the level of the circuit~$\cE^s_{B\to A'}$.
    \item 
    Unlike in the information-theoretic setting (see e.g.~\cite{KRS09OperationalMeaningEntropy,Tomamichel_2016}), the choice of basis for the maximally entangled $\ket{\Phi_{AA'}}$ state is highly relevant. In other words, different choices of basis result in differences in the circuit size needed to achieve the same fidelity. We choose the computational basis for both $A$ and $A'$.
    \item 
    There is a natural way to define the non-conditional computational min-entropy. For any state $\rho_{A}$, we say that
    \begin{equation*}
        \hmincompc[s](A)_{\rho} \coloneq \hmincompc[s](A|B)_{\rho} \;,
    \end{equation*}
    where $\rho_{AB}= \rho_{A} \otimes \ketbra{0}{0}_{B}$ and $B$ is a one-dimensional Hilbert space.
     We highlight that if $A$ is a classical register, then this unconditional entropy is close to the information-theoretic min-entropy. To see this, first note that the circuit which is applied to B simply generates the most likely classical output of $A$. In our setting, this is feasible, as generating any single bit-string can be done via an efficient circuit.  However, this argument only holds if $A$ is classical; it will not carry over to the fully quantum or bipartite settings.
     
     In other words, we care about the gate complexity of the applied circuit; we do not consider the computational cost that is needed to first find the optimal circuit.  As we discuss further on in this section, it is (in part) this subtle difference that causes the computational min-entropy to differ from the HILL entropy~\cite{CC17Computationalminentropy}.
    
\end{enumerate}

In line with standard practice in quantum information theory, we also introduce a smooth variant of the computational min-entropy (and later on the max-entropy).
\begin{definition}[Smooth Quantum Computational Min-Entropy]\label{def:intro:smooth_quantum_comp_entropy}
    For any sub-normalized bipartite quantum state $\rho_{AB}\in \cS_{\bullet}(AB)$, and $\varepsilon \ge 0$, $s \in \mathbb{N}$, we say that
    \begin{equation}
        \hmincompc(A|B)_{\rho} \coloneq \sup_{\tilde{\rho}\in\cB_{\varepsilon}(\rho)} \hmincompc[s](A|B)_{\tilde{\rho}} \;,
    \end{equation}
    where the optimization is over all states $\tilde{\rho}\in \cB_{\varepsilon}(\rho)$ that are $\epsilon$-close to the original state in purified distance.
\end{definition}
The smoothing step is crucial in operational settings, as it accounts for small errors or imperfections in the description of quantum states. This makes the entropies robust to noise and applicable in one-shot scenarios, where exact characterizations are too rigid to be meaningful, see e.g.~\cite{renner2006securityquantumkeydistribution,KRS09OperationalMeaningEntropy,Tomamichel_2016}.
Importantly, the smoothing is not computational; it only relaxes the description of the state within a small purified distance. Thus, while the computational power of the circuits matters for the entropy itself, the smoothing step is purely an analytic tool. The resulting smooth computational min- and max-entropies are therefore the natural objects for applications, just as their information-theoretic counterparts are in the standard setting.

With the definition in hand, we derive several fundamental properties of the (smooth) computational min-entropy. 
Among these, two consequences are particularly worth highlighting. For simplicity, we express them in terms of the (un-smoothed) computational min-entropy, but note that analogous results also hold for the smoothed quantity. First, the computational min-entropy satisfies a fully quantum leakage chain rule. We note that in contrast, the leakage chain rule for the HILL entropy~\cite{CC17Computationalminentropy} is currently known only in the restricted ccq setting, highlighting one advantage of our definition in the fully quantum setting.

\begin{lemma}[Fully Quantum Leakage Chain Rule]
    For any sub-normalized tripartite quantum state $\rho_{ABC}\in \cS_{\bullet}(ABC)$, and any $\varepsilon \ge 0$, $s\in\bbN$, $\ell = \log\dim(C)$, let $t_{\ell}$ be the size of a circuit generating $\omega_{C}$, the maximally mixed state. Then,
              \begin{equation*}
                  \hmincompc[(s,\varepsilon)](A|BC)_{\rho} \ge \hmincompc[(s+t_{\ell},\varepsilon)](A|B)_{\rho} - 2\ell \;.
              \end{equation*}
\end{lemma}
Second, because the computational min-entropy coincides with the unpredictability entropy, $\hunps[s]$, on cq-states, we can directly extend prior results on pseudo-randomness extraction for the latter~\cite{arnon2025computational} to the computational min-entropy. Succinctly phrased,  classical data with high computational min-entropy can be converted into bits that look uniform to any computationally bounded observer.

\begin{lemma}\label{lem:IPext_compmin_intro}
    Let $\rho_{XE}\in \cS_{\circ}(XE)$ be a normalized cq-state\footnote{Normalized simply means that $\tr{ \rho_{XE}}=1$. Moreover, cq-states are of the form $\rho_{XE}=\sum_{x} p_{x} \proj{x}\otimes \rho^{x}_{E}$, where $\rho^{x}_{E}$ are normalized quantum states.} where $X$ is distributed over $\bits^{n}$ and $Y$ be uniformly distributed over~$\bits^{n}$. Moreover, let $k_{\ext} \in \mathbb{N}, \varepsilon_{\ext}>0$ and $k_{\ext} \ge 1-2\log(\varepsilon_{\ext})$. We denote by~$\IP(X,Y)$ the binary inner-product of the values taken by $X$ and $Y$.
    If
    \begin{equation*}
        \hmincompc[(2s+3n+5)](X|E) \ge k_{\ext} \;,
    \end{equation*}
    then
    \begin{equation*}
        \dist_{s}(\rho_{\IP(X,Y)YE},\omega_{1}\otimes\rho_{YE}) \le \varepsilon_{\ext} \;.
    \end{equation*}
\end{lemma}
Here, $\dist_{s}$ denotes the s-computational distance, which captures an observer's ability to distinguish between two states, using at most $s$ gates. Below, we provide a brief list of properties that are satisfied by the computational min-entropy.
\begin{samepage}
\begin{restatable}{theorem}{IntroPropertiesCompMinEnt}\label{Lem:intro:PropertiesCompMinEnt}
    The smooth computational min-entropy, $\hmincompc(A|B)_{\rho}$, satisfies
    \begin{enumerate}
        \item (Monotonicity in $s$): For any $\rho_{AB}\in \cS_{\bullet}(AB)$, $\varepsilon \geq 0$, and $s^\prime \geq s$,
              \begin{align}
                  \hmincompc(A|B)_{\rho} \ge \hmincompc[\left(s',\varepsilon\right)](A|B)_{\rho}  \;.
              \end{align}
        \item (Monotonicity in $\varepsilon$): For any $\rho_{AB}\in \cS_{\bullet}(AB)$, $s \in \mathbb{N}$, and $\varepsilon^\prime \geq \varepsilon \geq 0$,
              \begin{align}
                  \hmincompc[\left(s,\varepsilon^\prime\right)](A|B)_{\rho} \ge \hmincompc[\left(s,\varepsilon\right)](A|B)_{\rho}  \;.
              \end{align}
        \item (Data Processing):  Let $\cC_{B\to B'}$ be a quantum channel that can be implemented using a circuit of size $t$. Then, for any $\rho_{AB}\in \cS_{\bullet}(AB)$, $s \in \mathbb{N}$, and $\varepsilon \geq 0$,
              \begin{align}
                  \hmincompc(A|B')_{\cC(\rho)} \ge \hmincompc[\left(s+t,\varepsilon\right)](A|B)_{\rho} \;.
              \end{align}
        \item (Leakage Chain Rule):     For any $\rho_{ABC}\in \cS_{\bullet}(ABC)$, and any $\varepsilon \ge 0$, $s\in\bbN$, denote by $d_{C}$ the dimension of the register $C$ let $t$ be the size of a circuit generating $\omega_{C}$, the maximally mixed state. Then,
              \begin{equation*}
                  \hmincompc(A|BC)_{\rho} \ge \hmincompc[(s+t,\varepsilon)](A|B)_{\rho} - 2\log(d_{C}) \;.
              \end{equation*}
        \item (Purification Chain Rule):  For any $\rho_{ABC}\in \cS_{\bullet}(ABC)$, and any $\varepsilon \ge 0$, $s\in\bbN$, denote by $d_{A}$ the dimension of the register $A$
              \begin{equation*}
                  \hmincompc(B|C)_{\rho} \le \hmincompc(AB|C)_{\rho} + \log(d_{A}) \;.
              \end{equation*}
        \item ($\hunps[s]$ Equivalence):   For any cq-state $\rho_{XE}\in \cS_{\bullet}(XE)$ such that  ${\ell = \log\dim(X) \in \mathbb{N}}$, and any $s\in\mathbb{N}$, $\varepsilon \ge 0$, 
              \begin{align}
                  \hmincompc[s](X|E)_{\rho} & = \hunps[s](X|E)_{\rho}                                                                     \\
                  \hmincompc(X|E)_{\rho}    & \geq \hunps[s]^{\varepsilon}(X|E)_{\rho} \geq \hmincompc[(s+l,\varepsilon)](X|E)_{\rho} \;.
              \end{align}
        \item ($\mathrm{H}_{\min}^{\varepsilon}$ Relation):  For any $\rho_{AB}\in \cS_{\bullet}(AB)$, $s \in \mathbb{N}$, and $\varepsilon \geq 0$,
              \begin{align}
                  \hmincompc(A|B)_{\rho}                   & \ge \mathrm{H}_{\min}^{\varepsilon}(A|B)_{\rho}    \\
                  \lim_{s\to\infty} \hmincompc(A|B)_{\rho} & = \mathrm{H}_{\min}^{\varepsilon}(A|B)_{\rho}  \;.
              \end{align}
    \end{enumerate}
\end{restatable}
\end{samepage}
One point of emphasis is that (i) there exists a strict separation between the computational min-entropy and its information-theoretic analog, and (ii) a similar separation exists between computational min-entropies of states that are (locally) isometrically equivalent. Both follow from a simple counting argument and do not rely on computational hardness assumptions.
\begin{restatable}{lemma}{IntroInfoSep}\label{Lem:intro:InfoSep}
    Let $s\in\mathbb{N}$ be fixed.
    There exists an integer $n$, and a normalized fully classical bipartite state ${\rho_{AB}\in \cS_{\circ}(AB)}$ such that
    \begin{equation*}
        \mathrm{H}_{\min}(A|B)_{\rho}=0
        \quad\text{and}\quad
        \hmincompc[s](A|B)_{\rho}\ge n/2 \;.
    \end{equation*}
\end{restatable}
In other words, there exist states for which $B$ contains all of the data stored in $A$, but this data cannot be adequately accessed with $s$ gates. Similar results can be obtained using one-way functions. Closely related to this, the following lemma shows that the computational min-entropy is not invariant under (local) isometric transformations. This further distinguishes the computational min-entropy from the information-theoretic min-entropy, which is invariant under such transformations~\cite{Tomamichel_2016}. In particular, it highlights that having access to the ``right" isometry can significantly enhance an adversary's ability to infer what is stored in $A$.
\begin{restatable}{lemma}{IntroAltPurMinEnt}\label{Lem:intro:AltPurMinEnt}
    Let $s\in\mathbb{N}$ be fixed.
    There exists an integer $n$, a normalized bipartite state ${\rho_{AB}\in \cS_{\circ}(AB)}$, and two purifications
    $\rho_{ABC}, \rho_{ABD}$
    of $\rho_{AB}$ such that
    \begin{equation*}
        \hmincompc[s](A|C)_{\rho}=-n
        \quad\text{and}\quad
        \hmincompc[s](A|D)_{\rho}\ge n/2 \;.
    \end{equation*}
\end{restatable}
We conclude the computational min-entropy overview with a comparison to another quantum computational entropy; the (relaxed) Hill entropy~\cite{CC17Computationalminentropy}.

\begin{definition} [HILL Entropy]
    A normalized bipartite quantum state $\rho_{AB}\in \cS_{\circ}(AB)$ has $(s,\varepsilon)$-relaxed HILL entropy at least $k$, i.e.\ $\mathrm{H}_{s,\varepsilon}^{\hill}(A|B)\ge k$,  if there is another state $\sigma_{AB}\in \cS_{\circ}(AB)$ such that $\textrm{H}_{\min}(A|B)_{\sigma} \ge k$ and $\dist_{s}(\rho_{AB},\sigma_{AB})\leq \varepsilon$.
\end{definition} 
This definition is conceptually very straightforward; if a state $\rho_{AB}$ is computationally indistinguishable from a state with high (information-theoretical) min-entropy, then it has high HILL entropy. One drawback of this entropy, however, is that the only known leakage chain rule solely holds for a restricted set of states, i.e. ccq-states $\rho_{XYC}$. That means, one can bound the entropy loss under quantum leakage only when the remaining data is classical. In contrast, our leakage chain rule holds for general states.

Let us now consider two examples that highlight that neither computational entropy dominates the other. For this, we introduce 
pseudo-random generators~\cite{Gold08}.

\begin{definition}[Pseudo-Random Generator (PRG)]
    A function $G: \bits^{n} \to \bits^{m}$ is called a $(s_{\prg},\varepsilon_{\prg})$ pseudo-random generator (PRG) if it is efficiently computable, length expanding ($m>n$), and the output distribution on a random input is $ \varepsilon_{\prg}$-computationally indistinguishable from the uniform distribution in the computational trace distance $\dist_{s_{\prg}}$.
\end{definition}
For PRGs with short seed, the HILL entropy of the generated distribution must be high. This is because the output is computationally indistinguishable from a uniformly random distribution, which has maximal (information-theoretic) min-entropy. However, the maximal min-entropy of the actual output distribution is determined by the seed length, $n$. Since the circuit that is needed to generate the best guess is efficient,\footnote{Again, note that we do not include the computational complexity that is necessary to determine what the optimal circuit is. This allows for a clean worst-case analysis of an adversary's capabilities.} we find that 
\begin{equation*}
    \mathrm{H}_{s_{\prg},\varepsilon_{\prg}}^{\hill}(\Prg(S))_{\sigma} \approx m \;, \quad \text{and} \quad  \hmincompc[\left(s_{\prg},\varepsilon_{\prg}\right)](\Prg(S))_{\sigma} \approx n \;.
\end{equation*}
One can choose $m$ to be polynomial in $n$~\cite{Gold08}, thus yielding a clear separation between the two quantities where the HILL entropy dominates.
Conversely, one-way permutations~\cite{kaliski1991onewayperm}, i.e. families of permutations $\set{f}_{f\in F}$ that can be efficiently computed but are computationally hard to invert (on average), provide an example in the other direction. 

For a random variable $X\in \bits^{n}$ and a family $\set{f}_{f\in F}$ of one-way permutations $f :\bits^{n}\to\bits^{n}$, the distribution of $(x,f(x),f)$, with uniform probability over $x\in X$ and choice of $f\in F$ would not have high HILL entropy $\mathrm{H}^{s,\varepsilon}_{\hill}(X|F(X),F)$. The joint conditional distribution $(x,f(x),f)$ is distinguishable from distributions where $x$ is not determined by $(f(x),f)$ for any adversary that can evaluate $f$. To distinguish between $(y,f(x),f)$ and $(x,f(x),f)$, all a distinguisher needs to do is to evaluate $f(y)$ or $f(x)$. Since $f$ is a one-way permutation, evaluating it in the forward direction is computationally easy. On the other hand, since $F$ is a family of \emph{one-way} permutations, guessing $x$ from $f(x),f$ is computationally hard, meaning that $\hmincompc(X|F(X),F)$ would be high. For such states, it would hold that
\begin{equation*}
    \mathrm{H}_{s_{\prg},\varepsilon_{\prg}}^{\hill}(X|F(X),F)_{\rho} < \hmincompc[\left(s_{\prg},\varepsilon_{\prg}\right)](X|F(X),F)_{\rho}\;.
\end{equation*}

\subsection*{Computational Max-Entropy}
Another important quantity in quantum information theory is the conditional max-entropy~\cite{KRS09OperationalMeaningEntropy,Tomamichel_2016}. It can be defined through the following duality relation with the conditional min-entropy.
\begin{definition}[Max-Entropy~\cite{KRS09OperationalMeaningEntropy}]\label{def:intro:max-entropy-math}
    Let $\rho_{AB} \in S_{\bullet}(AB)$ be a sub-normalized bipartite quantum state and let ${\rho_{ABC} \in S_{\bullet}(ABC)}$ denote a purification.\footnote{A purification is a state ${\rho_{ABC} \in S_{\bullet}(ABC)}$ that satisfies (i) $\tr{\rho_{ABC}}= \tr{\rho_{ABC}^2}$ (ii) tracing out $C$ yields the original state, i.e. ${\ptr{C}{\rho_{ABC}}=\rho_{AB}}$; see e.g.~\cite[Section 2.4.3]{Tomamichel_2016}.} The conditional max-entropy of $A$ given $B$ is defined as:
\begin{align*}
    \mathrm{H}_{\max}(A|B)_\rho \coloneq -\mathrm{H}_{\min}(A|C)_{\rho} \; .
\end{align*}
\end{definition}

In the information-theoretic setting, this duality can be taken as the very definition of the conditional max-entropy. It is well-defined because for any mixed state $\rho_{AB}$, any purification $\rho_{ABC}$ yields the same value of $\mathrm{H}_{\min}(A|C)_\rho$, since they are related by local isometries on $C$~\cite{KRS09OperationalMeaningEntropy,Tomamichel_2016}. In the computational setting, however, this invariance no longer holds: moving between different purifications may require operations of high computational complexity, and the resulting computational min-entropies can therefore differ. As a result, the duality alone does not provide a consistent definition of computational max-entropy.

To obtain a consistent notion in the computational setting, we therefore fix a specific purification to work with. In this work, we consider the pretty good purification~\cite{Winter_2004prettygoodpure}, which has a simple algebraic form and is widely used in quantum information theory, see e.g.~\cite{Winter_2004prettygoodpure,Iten_2017,arqand2024generalizedrenyientropyaccumulation}. By defining the computational max-entropy via a duality relation with the computational min-entropy on this particular purification, we arrive at a well-defined quantity that extends the information-theoretic max-entropy while incorporating computational constraints.

\begin{definition} [Quantum Computational Max-Entropy] \label{Def:intro:DualRel}
    Let $\rho_{AB}\in \cS_{\bullet}(AB)$ be a sub-normalized bipartite state.
    The quantum computational max-entropy is  defined as
    \begin{equation} 
        \hmaxcompc[s](A|B)_\rho \coloneq - \hmincompc[s](A|C)_{\rho^{\textrm{pg}}}\;,
    \end{equation}
where, 
\begin{align}
\rho^{\textrm{pg}}_{AC} =\ptr{B}{\proj{\rho^{\textrm{pg}}}_{ABC}} =
       \ptr{B}{\left(\sqrt{\rho_{A}} \otimes \mathbb{I}_C\right)
        \proj{\Phi}_{AB|C}
        \left(\sqrt{\rho_{A}} \otimes \mathbb{I}_C\right) }\; ,
\end{align}
 and $\ket{\Phi}_{AB|C}$ is the maximally entangled state between $AB$ and $C$.
\end{definition}
This quantity can also be smoothed, resulting in the following definition.
\begin{definition}[Smooth Quantum Computational Max-Entropy]\label{def:intro:smooth_quantum_max_entropy}
    For any sub-normalized bipartite quantum state $\rho_{AB}\in \cS_{\bullet}(AB)$, and $\varepsilon \ge 0$, $s \in \mathbb{N}$, we say that
    \begin{equation}
        \hmaxcompc(A|B)_{\rho} \coloneq \inf_{\tilde{\rho}\in\cB_{\varepsilon}(\rho)} \hmaxcompc[s](A|B)_{\tilde{\rho}} \;,
    \end{equation}
    where the optimization is over all states $\tilde{\rho}\in \cB_{\varepsilon}(\rho)$ that are $\epsilon$-close to the original state in purified distance.
\end{definition}

Moreover, with this choice of purification, we derive another closed-form expression for the computational max-entropy that depends only on the state $\rho_{AB}$.
\begin{figure}[ht]
    \centering
    \begin{quantikz}[row sep=0.3cm]
        \lstick[wires=8]{$\ket{\rho^{\textrm{pg}}_{AB|C}}$} & \qw & \qw \rstick[wires=2]{$A$} \\
        & \qw & \qw                       \\
        & \qw & \qw \rstick[wires=2]{$B$} \\
        & \qw & \qw                       \\
        & \gate[wires=4]{\mathcal{E}^s_{C \to A'}} & \qw \rstick[wires=4]{$A'$} \\
        &                                   & \qw                        \\
        &                                   & \qw                        \\
        &                                   & \qw
    \end{quantikz}
    \quad $\approx$ \quad $\proj{\Phi_{AA'}}\otimes \sigma_B$
    \caption{A circuit $\mathcal{E}^s_{C \to A'}$ acts on the purifying system $C$ for the pretty good purification of $\rho_{AB}$. The goal is to distill entanglement between $A$ and $C$ and decoupling from $B$, quantified by fidelity with $\proj{\Phi_{AA'}} \otimes \sigma_B$.}
    \Description[Entanglement distillation circuit]{Quantum circuit with eight input wires representing a purification of $\rho_{AB}$ conditioned on $C$. 
  A four-wire operation $\mathcal{E}^s$ acts on the $C$ subsystem and outputs four wires labeled $A'$. 
  The target state is close to a maximally entangled pair between $A$ and $A'$ tensor an arbitrary state on $B$.}

    \label{fig:intro:Es-circuit-max-entropy}
\end{figure}
\begin{lemma}[Closed-form definition for the computational max-entropy ]\label{Prop:nosmooth_quantum_comp_maxentropy_intro}
    For any bipartite quantum state $\rho_{AB}\in \cS_{\bullet}(AB)$, and $s \in \mathbb{N}$, it holds that
    \begin{equation*}
        \hmaxcompc[s](A|B)_{\rho} =  \log \max_{\sigma_{B}} \max_{\cE\in \cC(s)} d_{A} F\left((\mathbb{I}_{AB} \otimes \cE_{C\to A'} )   \proj{\rho^{\textrm{pg}}}_{ABC},\ketbra{\Phi_{AA'}}{\Phi_{AA'}}\otimes \sigma_{B} \right) \;,
    \end{equation*}
    where $\ket{\rho^{\textrm{pg}}}_{ABC} = 
        \left(\sqrt{\rho_{AB}} \otimes \mathbb{I}_C\right)
        \ket{\Phi}_{AB|C}$, and one is optimizing over all mixed quantum states $\sigma_{B}$ and quantum circuits $\cE^s_{C\to A'}$ acting on $C$, of size at most $s$.
\end{lemma}
This closely mirrors the alternative expression for the max-entropy found in, e.g.,~\cite{KRS09OperationalMeaningEntropy,Tomamichel_2016}. The main difference is that the register $C$ can be removed in the information-theoretic setting using Uhlmann's theorem~\cite{uhlmann1976transition}. However, since the required isometry is generally not expected to be efficiently preparable~\cite{bostanci2023unitary}, this simplification is not possible for the computational max-entropy.

By defining the computational max-entropy via a duality relation, it not only inherits several structural properties (discussed below), but is also tied to an operational task: for bounded circuit size $s$, it captures the maximal fidelity achievable when distilling entanglement between $A$ and the auxiliary register $C$ of the pretty good purification, in a way that decouples $A$ from $B$ (see also~\cref{fig:intro:Es-circuit-max-entropy}). 

Again, we note the following two features that naturally arise from the way we define the computational max-entropy.
\begin{enumerate}
    \item Analogous to the computational min-entropy, this definition holds for any initial state; we do not need to impose any conditions on the computational power needed to generate~$\rho_{AB}$. 
    \item 
    There is a natural way to define the non-conditional computational max-entropy. For any state $\rho_{A}$,  \begin{equation*}
        \hmaxcompc[s](A)_{\rho} \coloneq  \log \max_{\cE\in \cC(s)} d_{A} F\left((\sqrt{\rho_{A}} \otimes \cE_{C\to A'} )\ket{\Phi_{A|C}},\ket{\Phi_{AA'}} \right) \;.
    \end{equation*}
\end{enumerate}

For this definition of the computational max-entropy, we recover several intuitive properties: it is monotone in both the circuit-size and smoothing parameters, and converges to the information-theoretic smooth max-entropy as $s \to \infty$.
\begin{restatable}{theorem}{IntroPropertiesCompMaxEnt} \label{lem:intro:PropertiesCompMaxEnt}
    The smooth computational max-entropy, $\hmaxcompc(A|B)_{\rho}$, satisfies
    \begin{enumerate}
        \item (Monotonicity in $s$): For any $\rho_{AB}\in \cS_{\bullet}(AB)$, $\varepsilon \geq 0$, and $s^\prime \geq s$,
              \begin{align}
                  \hmaxcompc(A|B)_{\rho} \le \hmaxcompc[\left(s',\varepsilon\right)](A|B)_{\rho}  \;.
              \end{align}
        \item (Monotonicity in $\varepsilon$): For any $\rho_{AB}\in \cS_{\bullet}(AB)$, $s \in \mathbb{N}$, and $\varepsilon^\prime \geq \varepsilon \geq 0$,
              \begin{align}
                  \hmaxcompc[\left(s,\varepsilon^\prime\right)](A|B)_{\rho} \le \hmaxcompc[\left(s,\varepsilon\right)](A|B)_{\rho}  \;.
              \end{align}
        \item ($\mathrm{H}_{\max}^{\varepsilon}$ Relation):  For any $\rho_{AB}\in \cS_{\bullet}(AB)$, $s \in \mathbb{N}$, and $\varepsilon \geq 0$,
              \begin{align}
                  \hmaxcompc(A|B)_{\rho}                   & \le \mathrm{H}_{\max}^{\varepsilon}(A|B)_{\rho}\;,  \\
                  \lim_{s\to\infty} \hmaxcompc(A|B)_{\rho} & = \mathrm{H}_{\max}^{\varepsilon}(A|B)_{\rho}  \;.
              \end{align}
    \end{enumerate}
\end{restatable}

\section{Preliminaries}

\subsection{Basic Quantum Notation}
We briefly present the definitions and tools from quantum information theory that are needed for this work.
For a more general introduction to quantum information theory, we refer the reader to one of several books on the subject, such as~\cite{nielsen2010quantumbook,khatri2024principlesquantumcommunicationtheory,wilde2013quantumbook,Tomamichel_2016,renes_quantum_2022}.
In this work, we are interested in quantum states as they relate to information theory and computational tasks.
To represent quantum states, we work in finite-dimensional Hilbert spaces. A vector in a Hilbert space is denoted by $\ket{\phi}$ and its dual vector by~$\bra{\phi}$.
\begin{definition}
    A quantum state, $\rho$, is a positive semi-definite matrix with $\tr{\rho} \leq 1$.
    A pure quantum state is a state with a matrix rank of $1$. Pure states can be written in the form $\ket{\phi}\bra{\phi}$ for some vector $\ket{\phi}$. States that are not pure are called mixed states.
    Any mixed state can be written as a convex combination of pure states, i.e.,
    \begin{equation*}
        \rho = \sum_{i} p_{i} \proj{\rho_{i}} \;,
    \end{equation*}
    where $\set{p_{i}}$ is a probability distribution and $\proj{\rho_{i}}$ are pure states.
    We say a state is normalized if $\tr{\rho} = 1$ and sub-normalized if $\tr{\rho} \le 1$. We denote the sets of all normalized states on a Hilbert space $\cH_{A}$ by $\cS_{\circ}(A)$ and the set of all sub-normalized states by $\cS_{\bullet}(A)$. In cases where the space is clear from context, we will omit it and simply write $\cS_{\circ}$ or $\cS_{\bullet}$.
\end{definition}

Classical states can be viewed as a special case of quantum states.
Classical-quantum (cq) states arise when classical information is embedded in a quantum system and correlated with a quantum part.
Cq-states are a natural hybrid that appears frequently in quantum information theory and quantum cryptography.
\begin{definition}[Classical-Quantum (cq) State]
    A classical-quantum (cq) state is a state of the form
    \begin{equation*}
        \rho_{XE}=\sum_{x} p_{x} \proj{x}\otimes \rho^{x}_{E} \;,
    \end{equation*}
    with $\set{\ket{x}}_x$  being the standard basis vectors in $\cH_{X}$, representing the classical values of~$X$.
\end{definition}

One particularly useful classical distribution is the uniform distribution over a finite set of outcomes. In the quantum setting, we sometimes call it the maximally mixed state and write it as follows.
\begin{definition}[Maximally Mixed States]\label{def:max_mix}
    We say a state is maximally mixed if it is of the form
    \begin{equation*}
        \omega = \frac{1}{d} \sum_{i=1}^{d} \proj{i} \;,
    \end{equation*}
    where $\ket{i}$ is the standard basis of the Hilbert space of dimension $d$.
\end{definition}

A state is said to be bipartite or multipartite if the Hilbert space it is acting on is a tensor product space of two or more Hilbert spaces. We denote the Hilbert space of a system $A$ by $\cH_{A}$, the Hilbert space of a system $B$ by $\cH_{B}$, and the Hilbert space of the composite system $AB$ by $\cH_{AB} \coloneq \cH_{A} \otimes \cH_{B}$.
Similarly, for the states themselves, the subscript indicates the Hilbert space, e.g., $\rho_{AB}\in \cS_{\bullet}(AB)$.

In quantum information theory, entanglement is a fundamental resource that enables phenomena such as quantum teleportation~\cite{bennett1993teleporting}, superdense coding~\cite{bennett1992communicationsuperdenscoding}, and quantum key distribution~\cite{renner2006securityquantumkeydistribution}. Among entangled states, maximally entangled states play a central role, exhibiting the strongest possible correlations allowed by quantum mechanics, and serve as key resources in quantum information processing~\cite{Horodecki_2009_entanglement_measures}.

\begin{definition}[Maximally Entangled States]
    We say a bipartite state is maximally entangled if it is of the form $\proj{\Phi}$ where
    \begin{equation*}
        \ket{\Phi} = \frac{1}{\sqrt{d}} \sum_{i=1}^{d} \ket{i} \otimes \ket{i} \;,
    \end{equation*}
    where $\set{\ket{i}}_{i=1}^{d}$ is the standard computational basis of a space of dimension $d$.
\end{definition}

The following definitions formalize the different ways in which quantum states can be manipulated and accessed. Channels describe how quantum systems evolve in time or are transformed, while measurements specify how (classical) information can be extracted from them. Together, these concepts capture the allowed interactions with quantum states in quantum theory.

\begin{definition}[Quantum Channels]
    A quantum channel is a completely positive trace-preserving (CPTP) map. Channels map quantum states to quantum states.
    We use the following notation for channels acting on certain marginals of a state:
    \begin{equation*}
        \rho_{\phi(A)B} \coloneq (\phi_{A} \otimes \mathbb{I}_{B}) (\rho_{AB})\;,
    \end{equation*}
    to denote the state of the system after applying the channel $\phi$ on the marginal~$\rho_{A}$.
\end{definition}

\begin{definition}[POVM]
    Positive Operator-Valued Measures (POVMs) are generalized measurements that can be performed on quantum states.
    POVMs can be modeled as a set of positive semi-definite Hermitian matrices $\set{E_{i}}$ such that $\sum_{i}E_{i} = \1$. The probability of outcome $i$ on a state $\rho$ is $\tr{E_{i}(\rho)}$.
\end{definition}

Any quantum channel followed by any measurement can be modeled as a POVM~\cite{hellwigOperationsmeasurementsII1970}.
A POVM can be modeled by a channel operating on the state and some auxiliary system, followed by a projective measurement~\cite{stinespring1955dilationpositive}.

\subsection{Distance Measures}

We now present a few distance measures that we use throughout the paper.
We begin with the trace distance, a quantum analogue of statistical total variation distance.

\begin{definition}[Trace Distance]
    The trace distance between two normalized quantum states $\rho_{A},\sigma_{A} \in S_{\circ}\left(A\right)$ is defined as
    \begin{equation*}
        \dist(\rho_{A},\sigma_{A}) \coloneq \frac{1}{2} \norm{\rho_{A} - \sigma_{A}}_{1} = \frac{1}{2} \tr{\sqrt{(\rho_{A} - \sigma_{A})^{\dagger}(\rho_{A} - \sigma_{A})}} \;.
    \end{equation*}
\end{definition}
The trace distance between two quantum states $\rho_{A}$ and $\sigma_{A}$ quantifies their distinguishability: it equals the maximum difference in probabilities that these states can yield for the same measurement outcome.
The trace distance between states $\rho_{A},\sigma_{A}\in\cS_{\circ}(A)$ is the maximal probability of any POVM element to distinguish between them,~\cite[Theorem 9.1]{nielsen2010quantumbook}
\begin{equation*}
    \dist(\rho_{A},\sigma_{A}) = \max_{0\le P\le \1}(\tr{P(\rho_{A}-\sigma_{A})})\;.
\end{equation*}

The fidelity between quantum states is a central concept in quantum information theory, particularly in the context of entanglement and quantum state discrimination. It provides a measure of how distinguishable two quantum states are, and it plays a crucial role in the study of quantum entropies and their properties.
The trace distance is a natural quantum analogue of the classical total variation distance, quantifying how well two quantum states can be distinguished by measurements. However, in many quantum information tasks, especially those involving smoothing or operational one-shot scenarios~\cite{KRS09OperationalMeaningEntropy,Dupuis_2014_one_shot_decup}, the trace distance is not always the most suitable metric. This is because any mixed quantum state can be viewed as a reduced part of a larger pure state, and operational tasks often involve considering purifications and their overlaps.

Fidelity-based distance measures, such as the purified distance, are particularly well-suited for these settings. Due to the characterization provided by Uhlmann's theorem~\cite{uhlmann1976transition}, fidelity naturally captures the notion that two states are close if their purifications can be made to have a large overlap, reflecting the structure of quantum information as embedded in larger systems. This perspective motivates the use of fidelity and purified distance as the primary metrics for quantifying the closeness of quantum states in modern quantum information theory~\cite{Tomamichel_2010_Duality}. We will work with the following definition of the generalized fidelity, which is also well suited to work with sub-normalized states.

\begin{definition}[Generalized Fidelity]\label{def:GeneralizedFidelity}
    The (generalized) fidelity between two sub-normalized states ${\rho_{A},\sigma_{A} \in S_{\bullet}(A)}$ is determined by the maximal overlap between purifications:
    \begin{equation*}
        F(\rho_{A},\sigma_{A}) \coloneq \left(\max_{\ket{\phi_{AB}},\ket{\psi_{AB}}} \abs{\braket{\phi_{AB}}{{\psi_{AB}}}}+ \sqrt{(1-\tr{\phi_{AB}})(1-\tr{\psi_{AB}})}\right)^{2}\;,
    \end{equation*}
    where the maximum is taken over all pure states such that
    $\rho_{A}=\ptr{B}{\proj{\phi_{AB}}}$, and  $\sigma_{A}=\ptr{B}{\proj{\psi_{AB}}}$.
\end{definition}
Note that if at least one of the states is normalized, the generalized fidelity reduces to the known expression for the fidelity from~\cite{uhlmann1976transition}
    \begin{equation*}
        F(\rho_{A},\sigma_{A}) = \max_{\ket{\phi_{AB}},\ket{\psi_{AB}}} \abs{\braket{\phi_{AB}}{{\psi_{AB}}}}^{2}\;.
    \end{equation*}
In the context of quantum entropies, the purified distance arises as the most natural choice of metric, due to its compatibility with the operational tasks and smoothing techniques central to one-shot information theory~\cite{Tomamichel_2010_Duality,Tomamichel_2016}.
\begin{definition}[Purified Distance]\label{def:pur_dist}
    The purified distance between two sub-normalized states ${\rho_{A},\sigma_{A} \in S_{\bullet}(A)}$ is given by:
    \begin{equation*}
        \bP(\rho_{A},\sigma_{A}) \coloneq \sqrt{1 - F(\rho_{A},\sigma_{A})}  \;.
    \end{equation*}
\end{definition}

\begin{definition}[Purified Distance Ball]
    For a sub-normalized state $\rho\in \cS_{\bullet}(A)$ with ${\sqrt{\tr{\rho_{A}}}>\varepsilon\ge0}$, we define the $\varepsilon$-ball around $\rho$ as:
    \begin{equation*}
        \cB_{\varepsilon}(\rho_{A}) \coloneq \set{\tilde{\rho}_{A}\in \cS_{\bullet}(A): \; \bP(\rho_{A},\tilde{\rho}_{A})\le\varepsilon}  \;.
    \end{equation*}
\end{definition}

A central property of the purified distance that we utilize in this work is the following extension property.
\begin{lemma}[Extension Property for Purified Distance~\cite{Tomamichel_2010_Duality}]\label{lem:extension_uhlmann_property_for_purefied_distance}
    For any sub-normalized states $\rho_{AB},\sigma_{A}$ there is an extension, $\sigma_{AB}$, of $\sigma_{A}$, i.e. $\ptr{B}{\sigma_{AB}}=\sigma_{A}$, such that the purified distance remains the same:
    \begin{equation*}
        \bP(\rho_{AB},\sigma_{AB}) = \bP(\rho_{A},\sigma_{A}) \;.
    \end{equation*}
\end{lemma}

The extension property notably fails for quantum states when using the trace distance. If we further require that the extensions are both pure, the purified distance is essentially the only distance measure for which this property holds. For a more detailed discussion of the extension property of the purified distance, see~\cite[Appendix~H]{tan2024prospectsdeviceindependentquantumkey}.

\subsection{Entropies}
Entropies are fundamental quantities in information theory and cryptography, used to quantify the uncertainty or randomness of a system. Here, we define the min-entropy, an entropy measure that is related to tasks such as decoupling and privacy amplification, along with its smoothed counterpart~\cite{KRS09OperationalMeaningEntropy,Dupuis_2014_one_shot_decup}.

\begin{definition}[Min-Entropy~\cite{renner2006securityquantumkeydistribution}]\label{def:min-entropy-math}
    Let $\rho_{AB} \in S_{\bullet}(AB)$ be a sub-normalized bipartite quantum state. The conditional min-entropy of $A$ given $B$ is defined as:
    \begin{equation*}
        \mathrm{H}_{\min}(A|B)_{\rho} \coloneq  \sup_{\lambda \in \mathbb{R} ,\sigma_{B} \in \cS_{\bullet}(B) } \left\{ \lambda : \rho _{AB} \le 2^{-\lambda} (I_A \otimes \sigma_B) \right\}\;.
    \end{equation*}
\end{definition}

\begin{definition}[Smooth Min-Entropy~\cite{RennerKonig2005_smoothentext}]
    Let $\rho_{AB} \in S_{\bullet}(AB)$ be a sub-normalized bipartite quantum state and $\varepsilon\ge 0$. The conditional $\varepsilon$-smooth min-entropy of $A$ given $B$ is defined as:
    \begin{equation*}
        \mathrm{H}_{\min}^{\varepsilon}(A|B)_{\rho} \coloneq \sup_{\tilde{\rho}\in\cB_{\varepsilon}(\rho)}\mathrm{H}_{\min}(A|B)_{\tilde{\rho}}\;.
    \end{equation*}
\end{definition}
We now turn to the max-entropy, a dual quantity to min-entropy, which captures the effective support size of $A$ conditioned on $B$ and appears in the analysis of tasks such as data compression, state merging, and entanglement distillation~\cite{wild17_ent_distilation_via_decup,KRS09OperationalMeaningEntropy}.

\begin{definition}[Max-Entropy~\cite{KRS09OperationalMeaningEntropy}]\label{def:max-entropy-math}
    Let $\rho_{AB} \in S_{\bullet}(AB)$ be a sub-normalized bipartite quantum state. The conditional max-entropy of $A$ given $B$ is defined as:
    \begin{equation*}
        \mathrm{H}_{\max}(A|B)_{\rho} \coloneq  \sup_{\sigma_{B} \in \cS_{\circ}(B) } \log F(\rho_{AB},\1_{A}\otimes \sigma_{B}) \;.
    \end{equation*}
\end{definition}
\begin{definition}[Smooth Max-Entropy~\cite{Tomamichel_2010_Duality}]
    Let $\varepsilon\ge 0$ and $\rho_{AB} \in S_{\bullet}(AB)$ be a sub-normalized  bipartite quantum state. The conditional $\varepsilon$-smooth max-entropy of $A$ given $B$ is defined as:
    \begin{equation*}
        \mathrm{H}_{\max}^{\varepsilon}(A|B)_{\rho} \coloneq \inf_{\tilde{\rho}\in\cB_{\varepsilon}(\rho)}\mathrm{H}_{\max}(A|B)_{\tilde{\rho}}\;.
    \end{equation*}
\end{definition}

Smooth entropies generalize the conditional quantum entropies by permitting slight errors in the state. By setting $\varepsilon=0$ in the smooth entropies, we recover the non-smooth variants. 
\begin{equation*}
    \mathrm{H}_{\min}^{0}(A|B)_{\rho} = \mathrm{H}_{\min}(A|B)_{\rho}\;, \quad \mathrm{H}_{\max}^{0}(A|B)_{\rho} = \mathrm{H}_{\max}(A|B)_{\rho} \;.
\end{equation*}

A duality relation between the smooth min and max-entropy was shown in~\cite[Lemma 16]{Tomamichel_2010_Duality}.
\begin{lemma}\label{lem:duality_for_smooth_info_min_max_entropy}
    For any pure state $\rho_{ABC}\in \cS_{\bullet}(ABC)$, and any $0\le\varepsilon < \sqrt{\tr{\rho_{ABC}}}$,
    \begin{equation*}
        \mathrm{H}_{\max}^{\varepsilon}(A|B)_{\rho} = - \mathrm{H}_{\min}^{\varepsilon}(A|C)_{\rho}\;.
    \end{equation*}
\end{lemma}

These entropies can be seen as different measures of distance from the uniform and independent distribution using different notions of distance. This can be seen most directly for classical distributions, with no conditioning register. For a classical distribution $\set{P_{x}}_{x\in X}$:
$$\mathrm{H}_{\min}(P)=-\log(\max_{x} P_{x}) = -\log \Vert P \Vert_{\infty}\;,$$
and
$$\mathrm{H}_{\max}(P)=\log\left(\sum_{x} \sqrt{P_{x}}\right)^{2} = 2\log \sum_{x}\sqrt{P_{x}} \;.$$
Here, the min-entropy is related to the largest weight of the distribution. As such, it characterizes the best guess one can make with regard to the one-shot output generated by such a distribution. Conversely, the max-entropy is in some sense more related to the spread of the distribution, and hence related to tasks such as data compression~\cite{RenerWolf04smoothreny}.

\subsection{Quantum Computational Model}\label{sec:pre_cm}

In this work, we are concerned with scenarios in which adversaries or information-processing agents are limited in their computational power. Unlike the information-theoretic setting, which assumes unrestricted access to all physically allowed operations, the computational setting restricts the set of feasible operations according to some notion of bounded computational power. In this section, we introduce the quantum computational model used throughout the paper and formalize the notion of bounded computational power in this setting.

To model computational limitations in the quantum setting, we use quantum circuits as our underlying computational model. In this framework, algorithms and physical processes are represented by finite sequences of elementary quantum gates. While the space of physically allowed operations is in principle very large, we focus on processes that can be implemented with limited computational resources. Specifically, we measure computational effort by the circuit size, the total number of gates comprising the circuit. Intuitively, smaller circuits correspond to weaker computational power, so restricting circuit size provides a natural and concrete way to formalize the notion of bounded computational power for information-processing agents or adversaries.

The notion of a universal gate set plays a central role in the circuit model of quantum computation. Intuitively, a gate set is universal if any unitary operation, and hence any quantum computation, can be approximated to arbitrary accuracy by a finite sequence of gates drawn from this set. In the quantum setting, universality ensures that our model of computation is expressive enough to capture all physically realizable quantum processes, up to arbitrarily small error. Note that for some unitary operations, the approximation may require a very large number of gates and thus not be feasible for circuits of bounded size despite the universality. Importantly, universality is a property of the gate set, not of any individual circuit. 

\begin{definition}[Universal Gate Set~\cite{nielsen2010quantumbook}]
    A finite set of quantum gates $\mathcal{G}$ is called universal if for every unitary operation $U$ acting on $n$ qubits and every $\varepsilon > 0$, there exists a quantum circuit composed only of gates from $\mathcal{G}$ that implements a unitary $V$ such that
    \begin{equation*}
        \max_{\ket{\psi}} \Vert (U - V) \ket{\psi} \Vert  \le \varepsilon \;,
    \end{equation*}
    where $\Vert \cdot \Vert $ denotes the 2-norm and the maximum is taken over all pure normalized input states.
\end{definition}

For simplicity, we assume the gates in $\mathcal{G}$ act on one or two qubits; see e.g.~\cite[Section 4.5.3]{nielsen2010quantumbook} for a universal gate set with this property.
The specific choice of one and two-qubit gates for the universal gate set is not all too important in the context of this work; for simplicity, we assume the set includes the CNOT, H, and Z gates.

\begin{definition}[Quantum Circuits~\cite{yao1993quantumcircuitcomplexity}]
    We fix a finite set of universal elementary quantum gates. A quantum circuit is a sequence of quantum gates, that act on qubits, and measurements in the computational basis on any subset of the qubits. A quantum circuit acts on an input state and arbitrarily many qubit ancillas set to $\ketbra{0}{0}$, and may measure any subset of the qubits of the resulting state in the computational basis. The size of a circuit is the number of gates in the circuit, i.e., adding ancillas and applying computational measurements is not included in the computational complexity.\footnote{\diff{We note that intermediate measurements can be deferred to the end using ancillas and CNOT gates~\cite[Chapter 4.4]{nielsen2010quantumbook}. Furthermore, any controlled unitaries $\sum \ketbra{i}{i}\otimes U_i$ that are applied after intermediate measurements must be efficient as a whole; it would not suffice to simply assume that the conditional circuit $U_i$ are efficient.} } We denote by $\cC(s)$, the set of all circuits that can be implemented with at most $s$ gates.
\end{definition}

\begin{remark}    
    Note that we always consider a fixed $s$ for the size of the quantum circuit. One should think of $s$ as being chosen after choosing the computation model, i.e., the elementary universal set of one- and two-qubit gates. Asymptotic behavior is not the main object here; rather, the focus is on computational \emph{power}, not computational \emph{complexity}. Working with a fixed number of gates $s$ facilitates a nuanced approach in establishing the foundations and tradeoffs of quantum computational information theory. Considering asymptotic, non-uniform behaviors (see e.g.~\cite{yao1993quantumcircuitcomplexity}) in future potential applications is straightforward, since the results are stated for any fixed $s$. 
\end{remark}

Classical distributions are a special case of quantum states. We call circuits that map quantum states to classical distributions guessing circuits. Note that any circuit can be made into a guessing circuit without increasing the number of gates, by composing it with a measurement in the computational basis.

\begin{definition}[Guessing Circuits]
    A channel $\cC : \cS_{\circ}(A) \to \mathsf{Distr}(\bits^{n})$ is called a \emph{guessing circuit}, where $\mathsf{Distr}(\bits^{n})$ denotes the set of probability distributions over $\bits^{n}$. 
    We denote the set of all guessing circuits acting on $A$ with a circuit size of at most $s$ by $\cP(s)$.
\end{definition}

\begin{definition}[Distinguisher]
    A channel $\cC : \cS_{\circ}(A) \to \mathsf{Distr}(\bits)$ is called a \emph{distinguisher}, where $\mathsf{Distr}(\bits)$ denotes the set of probability distributions over $\bits$. 
    We denote the set of all distinguishers with circuit size at most $s$ by $\cD_{s}$.
\end{definition}

\begin{restatable}[Computational Distance]{definition}{defcompdist}\label{def:comp_dist}
    Let $\rho_{A},\sigma_{A}\in\cS_{\circ}(A)$ and~$s\in\bbN$. The $s$-computational distance between $\rho_{A}$ and $\sigma_{A}$ is given by:
    \begin{equation*}
        \dist_s(\rho_{A},\sigma_{A}) \coloneq \max_{\substack{\cC \in \cD(s)}} \abs{\pr{\cC(\rho_{A})=1} - \pr{\cC(\sigma_{A})=1}} \;.
    \end{equation*}
    We say that $\rho_{A}$ and $\sigma_{A}$ are $(s,\varepsilon)$-computationally indistinguishable if $\dist_s(\rho_{A},\sigma_{A}) \le \varepsilon$.
\end{restatable}

\subsection{Quantum Computationally Hard Tasks}
This section briefly introduces basic concepts from computational complexity theory and their relevance to information theory. In particular, we discuss the importance of functions that are easy to evaluate but hard to invert, which are central to both classical and quantum cryptographic primitives. We direct the interested reader to the works~\cite{bernstein2017postqcrypt,talbot2006complexityandcrypto,yao1993quantumcircuitcomplexity,nielsen2010quantumbook,morimae2024quantumunpredictability} for a more in-depth treatment of these topics.

In principle, quantum mechanics allows for a vast range of physical processes. However, not all such processes can be realized efficiently using quantum circuits. Even fully fault-tolerant quantum computers, if and when they are built, will be subject to some fundamental limitations on what can be computed using bounded resources. These limitations reflect inherent boundaries of efficient computation; In particular, it is well known that approximating a generic $n$-qubit unitary to constant precision typically demands computational resources that scale exponentially with $n$~\cite{jiaHayHaystackExplicitExamplesExponentialQuantumCircuitComplexity2023}.
This realization that not all computations can be easily executed is central to our understanding of computational hardness in the quantum setting. Since classical functions can be embedded as special cases of quantum channels, such limitations apply to both classical and quantum computation.

In our setting, we consider quantum circuits composed of a finite number of gates taken from a fixed universal gate set. For any fixed circuit size $s$, there exist quantum channels that cannot be implemented or approximated by any quantum circuit of size $s$. This idea is discussed, for example, in~\cite[Section~4.5.5]{nielsen2010quantumbook}, and captures the basic intuition that most physically allowed transformations are too complex to be efficiently realizable.

To illustrate this idea concretely, we consider the task of inverting a permutation. For any fixed gate bound $s$, one can choose the input length $n$ large enough so that the number of $n$-bit permutations exceeds the number of quantum circuits of size $s$ by a large margin. As a result, there must exist a permutation whose inverse cannot be computed, or even guessed with noticeable probability, by any such circuit. The following lemma formalizes this observation and provides a simple, unconditional example of computational hardness within our model.

\begin{restatable}{lemma}{HardPermutations}
    \label{lem:counting_hard_invers_fixed_s}
    Let $G$ be a universal gate set consisting of $|G|$ gates, with gates acting on at most $m$ qubits, and let
    $s\in\mathbb{N}$ be fixed. Then there exists an $n^\prime$ such that for all $n \ge n^\prime$ there is a permutation
    \begin{equation*}
        f:\bits^{n}\to \bits^{n} \;,
    \end{equation*}
   for which every quantum circuit $C$ using at most $s$ gates from $G$ satisfies
    \begin{equation*}
        \Pr_{y\leftarrow\bits^{n}} \left[C(y)=f^{-1}(y)\right] \le 2^{-n/2} \;.
    \end{equation*}
\end{restatable}
For a detailed proof, see Appendix~\ref{Appendix_hard_permutations_counting_compleatness_proofs}. The proof uses a counting argument and the probabilistic method to show the existence of such functions, it does not give us a constructive way to find such functions.

Note that Lemma~\ref{lem:counting_hard_invers_fixed_s} does not require any bound on the size of circuits that evaluate $f$ in the forward direction. For applications, we would prefer to work with concrete functions that can be evaluated in the forward direction efficiently, and for which finding an inverse is computationally hard. Such functions are known as one-way functions and are considered to be a basic tool in computational complexity-based cryptography~\cite{impagliazzo1989onewayfunc}. The existence of such one-way functions would imply $\mathrm{P}\neq\mathrm{NP}$. It is currently an open question in computer science whether $\mathrm{P}\neq\mathrm{NP}$ or $\mathrm{P}=\mathrm{NP}$. As such, the existence of one-way functions is generally based on hardness assumptions, see e.g.~\cite{cryptoeprint:2022/278,regev2009latticesLWE}, which are standard in complexity-based cryptography.

It is also interesting to ask whether there are functions that can be evaluated efficiently in the forward direction by a classical computer, such that even a quantum computer would not be able to find an inverse. This requirement is often called post-quantum cryptography or post-quantum security. There are, in fact, problems, such as learning with errors~\cite{regev2009latticesLWE}, that are believed to be hard for quantum computers to solve, and are useful for generating cryptographic primitives such as pseudo-random functions~\cite{banerjee2012pseudorandom}.

\section{Quantum Computational Min- and Max-Entropies}\label{sec:comp_min_max_entropies}

In this section, we define \textit{conditional}  quantum computational entropies on bipartite states, $\rho_{AB}$, where the computational min-entropy and max-entropy are denoted by $\hmincompc[s](A|B)_{\rho}$ and $\hmaxcompc[s](A|B)_{\rho}$, respectively.
We later discuss how these quantities are reduced to the non-conditional variants when the system $B$ becomes trivial.

\subsection{Computational Min-Entropy}

Our first contribution in this work is a definition of a new entropy-- the quantum computational min-entropy.
Looking at the formal definition of the min-entropy given in Definition~\ref{def:min-entropy-math}, it is not clear how one would incorporate the computational aspect to derive a new meaningful quantity.
Instead, we start by looking at the \emph{operational meaning} of the quantum information-theoretic min-entropy.
In~\cite[Theorem 2]{KRS09OperationalMeaningEntropy}, it was proven that for any bipartite quantum state $\rho_{AB}$, the conditional min-entropy
\begin{equation}\label{eq:min-entropy-op}
    \mathrm{H}_{\min}(A|B)_{\rho} = -\log q_{\text{corr}}(A|B)_{\rho} \;,
\end{equation}
where,
\begin{equation}
    q_{\text{corr}}(A|B)_{\rho} \coloneq d_A \max_{\cE} F \left( (\mathbb{I}_A \otimes \cE)(\rho_{AB}), \proj{\Phi_{AA'}} \right) \;,
\end{equation}
with $d_A$ the dimension of $A$, $\ket{\Phi_{AA'}}$ the maximally entangled state, and the maximum is taken over all quantum operations $\cE$ mapping $B$ to $A'$.

With the above result of~\cite{KRS09OperationalMeaningEntropy} in mind, we can see Equation~\eqref{eq:min-entropy-op} as the one \emph{defining} the min-entropy. That is, the min-entropy can be defined by the maximum achievable singlet fraction using an operation $\cE$ acting on the subsystem $B$.
This points to a new approach to tackling the computational aspect. Rather than evaluating every possible quantum operation $\cE$, we now permit only those executed via a quantum circuit of bounded size $s$.
We define the quantum computational min-entropy $\hmincompc[s](A|B)_{\rho}$ as follows:\footnote{For all entropies, we use the letter $c$ to indicate that we discuss a computational entropy.}
\begin{definition}[Quantum Computational Min-Entropy]\label{def:nosmooth_quantum_comp_entropy}
    For any bipartite quantum state $\rho_{AB}\in \cS_{\bullet}(AB)$, and $s \in \mathbb{N}$, we say that
    \begin{equation*}
        \hmincompc[s](A|B)_{\rho} \coloneq -\log d_{A} \max_{\cE^s_{B\to A'} }  F\left( (\mathbb{I}_A \otimes \cE^s)(\rho_{AB}), \proj{\Phi_{AA'}} \right) \;,
    \end{equation*}
    where the optimization is over the set of all quantum channels from $B$ to $A'$ that can be implemented by circuits of size at most $s$. 
\end{definition}
We illustrate the computational task associated with this definition in~\cref{fig:Es-circuit-min-entropy}.

\begin{figure}[ht]
    \centering
    \begin{quantikz}[row sep=0.3cm]
        \lstick[wires=8]{$\rho_{AB}$} & \qw & \qw                                      & \qw \rstick[wires=4]{$A$}  \\
        & \qw & \qw                                      & \qw                        \\
        & \qw & \qw                                      & \qw                        \\
        & \qw & \qw                                      & \qw                        \\
        & \qw & \gate[wires=4]{\mathcal{E}^s_{B \to A'}} & \qw \rstick[wires=4]{$A'$} \\
        & \qw & \qw                                      & \qw                        \\
        & \qw & \qw                                      & \qw                        \\
        & \qw & \qw                                      & \qw
    \end{quantikz}
    \quad $\approx$ \quad $\proj{\Phi_{AA'}}$
    \caption{A circuit diagram for the operational meaning of the computational min entropy. A channel $\mathcal{E}^s$ (circuit of size~$s$) acting only on the $B$ register. The goal of the circuit is to maximize the fidelity of the resulting global state with the maximally entangled state $F(\rho_{AA'},\proj{\Phi_{AA'}})$.}
    \Description[Computational min entropy diagram]{A circuit diagram for the operational meaning of the computational min entropy. A channel $\mathcal{E}^s$ (circuit of size~$s$) acting only on the $B$ register. The goal of the circuit is to maximize the fidelity of the resulting global state with the maximally entangled state $F(\rho_{AA'},\proj{\Phi_{AA'}})$.}
    \label{fig:Es-circuit-min-entropy}
\end{figure}

A few remarks are in order.
\begin{enumerate}
    \item Only the computational power of the circuit $\cE^s_{B\to A'}$ is considered. The state~$\rho_{AB}$ itself does not need to be efficient to prepare in order for the entropy to be well-defined. However, in most applications, it would make sense to impose constraints on the computational power required to prepare the state, whether locally or non-locally~\cite{gutoski2007toward,arnonfriedman2023computationalentanglementtheory}. No modification in the definitions of the entropies is required.
    \item In contrast to the information-theoretic case, the choice of basis for the maximally entangled $\ket{\Phi_{AA'}}$ state is important, as different choices of basis would result in differences in the size of the circuit. We choose the computational basis for both $A$ and $A'$.
\end{enumerate}

Next, we aim to adopt a standard approach in quantum information theory by examining a smooth version of our proposed entropy. Smooth entropies are crucial for applications in quantum information theory. These are defined by evaluating the entropy on a state $\tilde{\rho}_{AB}$, which is close to the original state $\rho_{AB}$ as measured by the purified distance (see Definition~\ref{def:pur_dist}). The use of purified distance ensures that the approximation is done while maintaining the essential characteristics of the original definitions.

By smoothing, the entropies become more robust and practically useful in real-world scenarios.
One intuitive reason is that the quantum states in realistic scenarios are never described with perfect precision, and small errors are tolerated (as long as they are carefully understood and taken into account) in applications.
Bearing these insights in mind, it becomes abundantly clear that computational power remains irrelevant when it comes to the process of smoothing. The necessity of weighing computational factors applies solely to the circuit, or the specific task at hand, rather than to the smoothing process, which is conducted methodically as a distinct mathematical approximation.
We thus choose to perform the smoothing step with the purified distance here as well.

To this end, we introduce a smoothing parameter $\varepsilon \ge 0$ to extend Definition~\ref{def:nosmooth_quantum_comp_entropy} to its smooth counterpart. Then:

\begin{definition}[Smooth Quantum Computational Min-Entropy]\label{def:smooth_quantum_comp_entropy}
    For any bipartite quantum state $\rho_{AB}\in \cS_{\bullet}(AB)$, and $\varepsilon \ge 0$, $s \in \mathbb{N}$, we say that
    \begin{equation}
        \hmincompc(A|B)_{\rho} \coloneq -\log d_{A} \min_{\tilde{\rho}\in \cB_{\varepsilon}(\rho)}\max_{\cE^s \in \cC(s)}  F\left( (\mathbb{I}_A \otimes \cE^s)(\tilde{\rho}_{AB}), \proj{\Phi_{AA'}} \right) \;,
    \end{equation}
    where the optimization is over all states $\tilde{\rho}\in \cB_{\varepsilon}(\rho)$ and all quantum circuits $\cE^s_{B\to A'}$ acting on $B$, of size at most $s$.
\end{definition}

We point out that when the circuit size is unbounded, i.e., $s\rightarrow\infty$, we retrieve the information-theoretic min-entropy. That is, $\lim_{s\to\infty} \hmincompc(A|B)_{\rho}  = \mathrm{H}_{\min}^{\varepsilon}(A|B)_{\rho}$.

\subsubsection{Non-Conditional Entropy}
Although we usually consider bipartite systems and conditional entropies, our definitions can be straightforwardly applied to a single-party setting. Formally, our definitions require specifying a bipartite state. To define the computational entropy of a state $\rho_{A}$ by itself, $\hmincompc(A)_{\rho}$, we can instead consider the case when $B$ is empty. In such cases, one should simply consider the state $\rho_{AB}= \rho_{A} \otimes \ketbra{0}{0}_{B}$, where $\mathcal{H}_{B}$ is a one-dimensional Hilbert space spanned by $\{ \ket{0} \}$. Here, $B$ contains a trivial ancilla register, which can always be instantiated. 

\begin{definition}[Non-Conditional Quantum Computational Min-Entropy]
    For any quantum state $\rho_{A}\in \cS_{\bullet}(A)$, $\varepsilon \geq 0$, and $s \in \mathbb{N}$, we say that
    \begin{equation*}
        \hmincompc(A)_{\rho} \coloneq \hmincompc(A|B)_{\rho} \;,
    \end{equation*}
    where $\rho_{AB}= \rho_{A} \otimes \ketbra{0}{0}_{B}$ and $B$ is a one-dimensional Hilbert space.
\end{definition}

In certain cases, this unconditional entropy is close to the information-theoretic min-entropy. If $A$ is classical for example, then the circuit that is applied to B simply generates the most likely classical output of $A$. In our setting, this is feasible, as generating any single bit-string can be done via an efficient circuit. However, this argument only holds if $A$ is classical; it will not carry over to the fully quantum or bipartite settings.

\subsubsection{Classical-Quantum States}
Let us examine the case in which the state $\rho_{XB}$ is a classical-quantum (cq-) state, i.e. $X$ describes a classical register.
It is well known that in the information-theoretic setup, the min-entropy is determined by the maximal guessing probability of the value of $X$, while having access to the quantum system $B$~\cite{KRS09OperationalMeaningEntropy}. That is,
\begin{equation}
    \mathrm{H}_{\min}(X|B)_{\rho} = -\log(p_{\text{guess}}(X|B)_{\rho}) \;,
\end{equation}
where, for POVMs $\set{E_{B}^{x}}_{x}$,
\begin{equation}
    p_{\text{guess}}(X|B)_{\rho} = \max_{\set{E_{B}^{x}}_{x}} \sum_{x} p(x)\tr{E_{B}^{x}\rho_{B}^{x}} \;.
\end{equation}

With this operational meaning in mind, two of the current authors have recently defined the computational analog of the min-entropy, coined the \emph{quantum computational unpredictability entropy}~\cite{arnon2025computational}, studied its properties and relevance for quantum cryptography. The entropy is defined for cq-states only (while Definitions~\ref{def:nosmooth_quantum_comp_entropy} and~\ref{def:smooth_quantum_comp_entropy} are fully quantum), as follows:

\begin{definition}[Quantum Computational Unpredictability Entropy] \label{def:notsmooth_unpredictability_entropy}
    For any cq-state ${\rho_{XB}\in \cS_{\bullet}(XB)}$, and $s\in\bbN$.
    We say that
    \begin{equation*}
        \hunps[s](X|B)_{\rho} \coloneq -\log(\max_{\cE \in \cP(s)}   \pr{\cE(\rho_{B}^x) = x}) \; ,
    \end{equation*}
    where one is optimizing over $\cP(s)$, the set of all guessing circuits of size at most $s$.
\end{definition}

Similarly, it can be smoothed by considering the cq-states in the purified $\varepsilon$-ball around the state
\begin{definition}[Smooth Quantum Computational Unpredictability Entropy]\label{def:smooth_unpredictability_entropy}
    For any cq-state $\rho_{XB}\in \cS_{\bullet}(XB)$, and $\varepsilon \geq 0,s\in\bbN$.
    We say that
    \begin{equation*}
        \hunps[s]^{\varepsilon}(X|B)_{\rho} \coloneq -\log(\min_{\tilde{\rho}  \in \cB^{\mathrm{cq}}_{\varepsilon}(\rho) } \max_{\cE \in \cP(s)}   \pr{\cE(\tilde{\rho}_{B}^x) = x}) \; ,
    \end{equation*}
    where one is optimizing over $\cP(s)$, the set of all guessing circuits of size at most $s$, and $\cB^{\mathrm{cq}}_{\varepsilon}(\rho)$ the set of cq-states in purified distance at most $\varepsilon$ from $\rho_{AB}$.
\end{definition}

The quantum computational unpredictability entropy $\hunps[s](X|B)$ is both the quantum counterpart to the classical unpredictability entropy~\cite{HLR07SepPseudoentropyfromCompressibility} as well as the computational version of the quantum min-entropy~\cite{KRS09OperationalMeaningEntropy}.
Furthermore, it satisfies the properties typically associated with the quantum min-entropy, adapted to reflect the computational constraints, as shown in~\cite{arnon2025computational}.

The operation $\cE^s$ can be interpreted as a guessing strategy so that the probability of correctly guessing $X$ given $B$ when $X = x$ is given by $ \bra{x} \cE(\rho_B^x) \ket{x}$.
Thus, the quantum computational unpredictability entropy is of special relevance for cryptographic tasks in which the adversary holding the system $B$ has a limited computational power. Notably, it was shown in~\cite{arnon2025computational} that the proposed entropy can be used in the context of pseudo-randomness extraction with quantum side information.

The quantum computational min-entropy and quantum computational unpredictability entropy are closely related. For cq-states, we show in Theorem~\ref{Lem:PropertiesCompMinEnt} below that these two quantities are  equal,  $\hmincompc[s](X|E)_{\rho}=\hunps[s](X|E)_{\rho}$ as one would hope. Moreover, their \textit{smoothed} versions are equivalent up to minor differences in circuit size. The difference in the number of gates ($s$) is bounded by some fixed constant times the log of the dimension of $X$ (and is constant depth); see proofs in~\cref{app_subsec_cq-eqauive_proofs}.
As such, one should view $\hmincompc[s](A|B)$ as the fully quantum version of the quantum computational unpredictability entropy, which was only defined for cq-states.

\subsubsection{Comparison to Previous Work}

We conclude this section by discussing alternative notions of quantum computational min-entropy proposed in previous work~\cite{CC17Computationalminentropy}.

A complementary approach to defining computational entropy in the quantum setting was developed by~\cite{CC17Computationalminentropy}, where the authors adapt classical notions of HILL entropy~\cite{HILL1999pseudorandom} to the quantum domain. Their work focuses on quantum analogues of HILL and metric entropy, both of which aim to characterize computational entropy via indistinguishability from high min-entropy states. These definitions follow a different philosophical route than our operational approach: rather than bounding what a computationally restricted adversary can learn, they ask whether the state is computationally indistinguishable from an ideal one with high entropy.

The authors introduce four definitions of quantum pseudo-entropy for the conditional quantum setting: strict HILL, relaxed HILL, strict metric, and relaxed metric. The distinction between HILL and metric lies in the order of quantifiers. For HILL entropy, there must exist a single ideal state with high min-entropy that is indistinguishable from the given state to all computationally bounded distinguishers. In contrast, metric entropy allows for a different high min-entropy ideal state for each distinguisher. The difference between the strict and relaxed variants is related to constraints on the conditioning register, which are not particularly relevant to the discussion and therefore omitted.

For the relaxed HILL entropy, the authors establish a leakage chain rule for a restricted class of ccq states.
Formally, the relaxed HILL entropy is defined as follows~\cite[Definition 4.4]{CC17Computationalminentropy}
\begin{definition}
    A quantum state $\rho_{AB}\in \cS_{\circ}(AB)$ has $(s,\varepsilon)$-relaxed HILL entropy at least $k$, i.e.\ $\mathrm{H}_{s,\varepsilon}^{\hill}(A|B)\ge k$,  if there is another state $\sigma_{AB}\in \cS_{\circ}(AB)$ such that $\textrm{H}_{\min}(A|B)_{\sigma} \ge k$ and $\rho_{AB},\sigma_{AB}$ are $(s,\varepsilon)$ computationally indistinguishable.
\end{definition}

The relaxed HILL entropy satisfies the following leakage chain rule in the ccq case~\cite[Theorem 6.2.]{CC17Computationalminentropy}.

\begin{lemma}
    For any $n,m,\ell,s'\in \mathbb{N}$ and $\varepsilon>0$ the following holds for $s=\mathrm{poly}(s',n,2^\ell,1/\varepsilon)$ and $\varepsilon'=\mathrm{O}(\varepsilon)$. Let $\rho_{XZB}\in \cS_{\circ}(XZB)$ be a ccq state with $n+m$ classical bits and $\ell$ qubits. If~${\mathrm{H}_{s,\varepsilon}^{\hill}(X|Z) \ge k}$, then~$\mathrm{H}_{s',\varepsilon'}^{\hill}(X|ZB) \ge k-\ell$.
\end{lemma}

In contrast, for~\cref{def:nosmooth_quantum_comp_entropy}, we define computational entropy directly by limiting the computational power of the adversary attempting to learn information about the state, rather than via indistinguishability from ideal states.
Our notion of computational min-entropy fulfills a leakage chain rule in the fully quantum setting~\cref{lem:chain_rule_quantum_comp_entropy}, with only linear loss in the computational parameter and no loss in the smoothing error.

Under standard computational assumptions, we see that HILL entropy and the computational min entropy presented in this work are different, and neither entropy dominates the other. This can be seen from the following two examples. First, let us define pseudo-random generators, an illustrative object for the advantage of indistinguishability-based computational entropies.

\begin{definition}[Pseudo-Random Generator (PRG)]
    A function $G: \bits^{n} \to \bits^{m}$ is called a $(s_{\prg},\varepsilon_{\prg})$ pseudo-random generator (PRG) if it is efficiently computable, length expanding ($m>n$) and the output distribution on a random input is $(s_{\prg},\varepsilon_{\prg})$ computationally indistinguishable from the uniform distribution.
\end{definition}

Under standard cryptographic assumptions, PRGs exist, with security for both classical and quantum distinguishers~\cite{banerjee2012pseudorandom}. Usually, the definition of PRGs is asymptotic, defining an infinite family of functions parametrized by $n$. The efficiency and $m,s_{\prg},\varepsilon_{\prg}$ should be seen as functions of $n$ in that case.

A strength of indistinguishability-based entropy notions is their suitability for reasoning about some computational objects, such as pseudo-random generators. For example, HILL entropy assigns high pseudo-entropy to the output of a pseudo-random generator (PRG) applied to a short random seed, since this output is by definition computationally indistinguishable from a uniformly random string. However, under our definition, the output is not highly unpredictable, as an efficient adversary can guess a random seed and compute the PRG output using bounded computational resources. Formally, for a maximally mixed state $\rho_{S}$ and $\Prg:\bits^{n} \to \bits^{m}$ a $(s_{\prg},\varepsilon_{\prg})$ pseudo-random generator and $\sigma = \rho_{\Prg(S)}$. Since $\Prg$ can be evaluated using less then $s_{\prg}$ gates we see:

\begin{equation*}
    \mathrm{H}_{s_{\prg},\varepsilon_{\prg}}^{\hill}(\Prg(S))_{\sigma} \approx m \;, \quad \text{and} \quad  \hmincompc[\left(s_{\prg},\varepsilon_{\prg}\right)](\Prg(S))_{\sigma} \approx n \;.
\end{equation*}
Since $m>n$, this gives an example for states such that
\begin{equation*}
    \mathrm{H}_{s_{\prg},\varepsilon_{\prg}}^{\hill}(X)_{\rho} > \hmincompc[\left(s_{\prg},\varepsilon_{\prg}\right)](X)_{\rho}\;.
\end{equation*}

On the other hand, one-way permutations~\cite{kaliski1991onewayperm} provide an example in the other direction. One-way permutations are families of permutations $\set{f}_{f\in F}$ that can be efficiently computed but are computationally hard to invert on average with non-negligible probability.

For a random variable $X\in \bits^{n}$ and a family $\set{f}_{f\in F}$ of one-way permutations $f :\bits^{n}\to\bits^{n}$, the distribution of $(x,f(x),f)$, with uniform probability over $x\in X$ and choice of $f\in F$ would not have high HILL entropy $\mathrm{H}^{s,\varepsilon}_{\hill}(X|F(X),F)$. The joint conditional distribution $(x,f(x),f)$ is distinguishable from distributions where $x$ is not determined by $(f(x),f)$ for any adversary that can evaluate $f$. To distinguish between $(y,f(x),f)$ and $(x,f(x),f)$, all a distinguisher needs to do is to evaluate $f(y)$ or $f(x)$. Since $f$ is a one-way permutation, evaluating it in the forward direction is computationally easy. On the other hand, since $F$ is a family of \emph{one-way} permutations, guessing $x$ from $f(x),f$ is computationally hard, meaning $\hmincompc(X|F(X),F)$ would be high. This thus demonstrates that for such states
\begin{equation*}
    \mathrm{H}_{s_{\prg},\varepsilon_{\prg}}^{\hill}(X|F(X),F)_{\rho} < \hmincompc[\left(s_{\prg},\varepsilon_{\prg}\right)](X|F(X),F)_{\rho}\;.
\end{equation*}

\subsubsection{Properties of the Computational Min-Entropy}
The information-theoretic smooth min-entropy satisfies a range of useful mathematical properties~\cite{Tomamichel_2016}, as well as operational uses~\cite{KRS09OperationalMeaningEntropy}, making it an ideal quantity to work with. In this section, we show that several of these properties can be extended to the computational setting. Moreover, we include additional properties, which relate quantum computational min-entropy to both its information-theoretic counterpart as well as to the quantum computational unpredictability entropy. The following lemma provides a useful overview of our results. 

\begin{restatable}{theorem}{PropertiesCompMinEnt}\label{Lem:PropertiesCompMinEnt}
    The smooth computational min-entropy, $\hmincompc(A|B)_{\rho}$, satisfies
    \begin{enumerate}
        \item (Monotonicity in $s$): For any $\rho_{AB}\in \cS_{\bullet}(AB)$, $\varepsilon \geq 0$, and $s^\prime \geq s$,
              \begin{align}
                  \hmincompc(A|B)_{\rho} \ge \hmincompc[\left(s',\varepsilon\right)](A|B)_{\rho}  \;.
              \end{align}
        \item (Monotonicity in $\varepsilon$): For any $\rho_{AB}\in \cS_{\bullet}(AB)$, $s \in \mathbb{N}$, and $\varepsilon^\prime \geq \varepsilon \geq 0$,
              \begin{align}
                  \hmincompc[\left(s,\varepsilon^\prime\right)](A|B)_{\rho} \ge \hmincompc[\left(s,\varepsilon\right)](A|B)_{\rho}  \;.
              \end{align}
        \item (Data Processing):  Let $\cC_{B\to B'}$ be a quantum channel that can be implemented using a circuit of size $t$. Then, for any $\rho_{AB}\in \cS_{\bullet}(AB)$, $s \in \mathbb{N}$, and $\varepsilon \geq 0$,
              \begin{align}
                  \hmincompc(A|B')_{\cC(\rho)} \ge \hmincompc[\left(s+t,\varepsilon\right)](A|B)_{\rho} \;.
              \end{align}
        \item (Leakage Chain Rule):     For any $\rho_{ABC}\in \cS_{\bullet}(ABC)$, and any $\varepsilon \ge 0$, $s\in\bbN$, denote by $d_{C}$ the dimension of the register $C$ let $t$ be the size of a circuit generating $\omega_{C}$, the maximally mixed state. Then,
              \begin{equation*}
                  \hmincompc(A|BC)_{\rho} \ge \hmincompc[(s+t,\varepsilon)](A|B)_{\rho} - 2\log(d_{C}) \;.
              \end{equation*}
        \item (Purification Chain Rule):  For any $\rho_{ABC}\in \cS_{\bullet}(ABC)$, and any $\varepsilon \ge 0$, $s\in\bbN$, denote by $d_{A}$ the dimension of the register $A$
              \begin{equation*}
                  \hmincompc(B|C)_{\rho} \le \hmincompc(AB|C)_{\rho} + \log(d_{A}) \;.
              \end{equation*}
        \item ($\hunps[s]$ Equivalence):   For any cq-state $\rho_{XE}\in \cS_{\bullet}(XE)$ such that  ${\ell = \log\dim(X) \in \mathbb{N}}$, and any $s\in\mathbb{N}$, $\varepsilon \ge 0$, 
              \begin{align}
                  \hmincompc[s](X|E)_{\rho} & = \hunps[s](X|E)_{\rho}                                                                     \\
                  \hmincompc(X|E)_{\rho}    & \geq \hunps[s]^{\varepsilon}(X|E)_{\rho} \geq \hmincompc[(s+l,\varepsilon)](X|E)_{\rho} \;.
              \end{align}
        \item ($\mathrm{H}_{\min}^{\varepsilon}$ Relation):  For any $\rho_{AB}\in \cS_{\bullet}(AB)$, $s \in \mathbb{N}$, and $\varepsilon \geq 0$,
              \begin{align}
                  \hmincompc(A|B)_{\rho}                   & \ge \mathrm{H}_{\min}^{\varepsilon}(A|B)_{\rho}    \\
                  \lim_{s\to\infty} \hmincompc(A|B)_{\rho} & = \mathrm{H}_{\min}^{\varepsilon}(A|B)_{\rho}  \;.
              \end{align}
    \end{enumerate}
\end{restatable}
For a proof of Lemma~\ref{Lem:PropertiesCompMinEnt}, we refer the reader to Appendix~\ref{App:Proofs}.

The next lemma demonstrates a sharp separation between the information-theoretic and computational conditional min-entropy, even for fully classical states. 

\begin{restatable}{lemma}{InfoCompMinSep}\label{Lem:InfoCompMinSep}
    Let $s\in\mathbb{N}$ be fixed.
    There exists an integer $n$, and a normalized fully classical bipartite state ${\rho_{AB}\in \cS_{\circ}(AB)}$ such that
    \begin{equation*}
        \mathrm{H}_{\min}(A|B)_{\rho}=0
        \quad\text{and}\quad
        \hmincompc[s](A|B)_{\rho}\ge n/2 \;.
    \end{equation*}
\end{restatable}
For the proof of Lemma~\ref{Lem:InfoCompMinSep}, we refer the reader to Appendix~\ref{App:Proofs}.

\subsection{Computational Max-Entropy} \label{Sec:CompMaxEnt}
The dual of the (smooth) information-theoretic min-entropy is known as the max-entropy~\cite{Tomamichel_2010_Duality}, and it plays a central role in tasks such as data compression~\cite{KRS09OperationalMeaningEntropy} and entanglement distillation~\cite{khatri2024principlesquantumcommunicationtheory}. Our second main contribution is to define a notion of \emph{quantum computational max-entropy}, inspired by the duality that relates the information-theoretic min- and max-entropies.

In the information-theoretic setting, the duality relation asserts that for any pure state $\rho_{ABC}$, one has $\mathrm{H}_{\max}(A|B)_{\rho}=-\mathrm{H}_{\min}(A|C)_{\rho}$, see Lemma~\ref{lem:duality_for_smooth_info_min_max_entropy}. Crucially, this yields a well-defined max-entropy because the min-entropy $\mathrm{H}_{\min}(A|C)_{\rho}$ is invariant under the choice of purification of $\rho_{AB}$. All purifications of a given mixed state $\rho_{AB}$ are related by local isometries on the purifying system $C$, and the min-entropy is unchanged under such isometries.

In the computational setting, however, this invariance breaks down. The value of the computational min-entropy can vary significantly between purifications of the same $\rho_{AB}$, since the isometries relating them may require substantial computational resources to implement. As a result, the straightforward application of the duality relation does not yield a well-defined notion of computational max-entropy. This issue is illustrated by the following lemma, proving the existence of pairs of purifications for the same state with drastically different computational min-entropy. The construction is based on the existence of permutations that are hard to invert for circuits of size at most $s$. The existence of permutations with this property follows from a standard counting argument, once the number of permutations in a set is large enough, no circuit with $s$ gates can approximate all of them. This argument holds for both classical and quantum circuits. For completeness, we include a formal statement and proof of the counting argument with the relevant parameters to the following lemma in Appendix~\ref{Appendix_hard_permutations_counting_compleatness_proofs}.

\begin{restatable}{lemma}{AltPurMinEnt}\label{Lem:AltPurMinEnt}
    Let $s\in\mathbb{N}$ be fixed.
    There exists an integer $n$, a bipartite state ${\rho_{AB}\in \cS_{\circ}(AB)}$, and two purifications
    $\rho_{ABC}, \rho_{ABD}$
    of $\rho_{AB}$ such that
    \begin{equation*}
        \hmincompc[s](A|C)_{\rho}=-n
        \quad\text{and}\quad
        \hmincompc[s](A|D)_{\rho}\ge -n/2
    \end{equation*}
\end{restatable}

\begin{proof}
    For any $s\in\mathbb{N}$, we choose $n$ such that Lemma~\ref{lem:counting_hard_invers_fixed_s} holds. Then, let
    \begin{equation*}
        \rho_{A} = \frac{1}{2^{n}}\sum_{x\in\bits^{n}}\ketbra{x}{x},
        \qquad \rho_{B} \coloneq \proj{0}.
    \end{equation*}
    Since $B$ is trivial, every purification of $\rho_{A}$ is admissible.
    Let us first define a purification with minimal computational min-entropy. Define
    \begin{equation*}
        \ket{\Phi_{AC}}
        =
        \frac{1}{\sqrt{2^{n}}}
        \sum_{x\in\bits^{n}}\ket{x}_{A}\ket{x}_{C},
        \qquad
        \rho_{ABC} \coloneq \ketbra{\Phi_{AC}}{\Phi_{AC}} \otimes \rho_{B}.
    \end{equation*}
    This state is already maximally entangled. Therefore, the identity channel on $C$, only relabeling it to $A'$ is a valid circuit of size less than $s$, demonstrating that
    \begin{equation*}
        \hmincompc[s](A|C)_{\rho} \leq -n \;.
    \end{equation*}
   Equality then holds due to Lemma~\ref{Lem:PropertiesCompMinEnt}, i.e.
    \begin{align}
        -n = \mathrm{H}_{\min}(A|C)_\rho \leq \hmincompc[s](A|C)_{\rho} \leq -n \; .
    \end{align}
    By Lemma~\ref{lem:counting_hard_invers_fixed_s} there exists a permutation $f:\bits^{n}\to\bits^{n}$
    such that every quantum circuit of size at most $s$ inverts $f$ with success probability at most $2^{-n/2}$.
    We set
    \begin{equation*}
        \ket{\psi_{AD}}=
        \frac{1}{\sqrt{2^{n}}}
        \sum_{x\in\bits^{n}}\ket{x}_{A}\ket{f(x)}_{D},
        \qquad
        \rho_{ABD} \coloneq \ketbra{\psi_{AD}}{\psi_{AD}} \otimes \rho_{B}.
    \end{equation*}
     For any channel $\cE^{s}_{D\to A'}$ we can define a guessing strategy that applies $\cE^{s}_{D\to A'}$ and measures $A'$ in the computational basis. Since measuring in the computational basis is free in this computational model, the guessing strategy uses the same number of gates as the channel, which we denote as this guessing strategy $\hat{\cE}^{s}$. \diff{Additionally, we define the operation $\Pi_{XX}$ which simply acts as
     \begin{align}
         \Pi_{XX} \left( \cdot \right) \coloneq \sum_{x} \proj{xx} \left(\cdot \right) \proj{xx} \; .
     \end{align}
Applying this projection on $\ket{\Phi_{AA'}}$ is equivalent to measuring both registers in the computational basis, i.e.,
     \begin{align}\label{Eq: PhiProjRel}
         \Pi_{XX}\left(\proj{\Phi_{AA'}}\right) = \sum_{x,y} \proj{xy}\proj{\Phi_{AA'}} \proj{xy}\; ,
     \end{align}
     as the measurement outputs of $A$ and $A'$ always coincide for this maximally entangled state.
     The Fidelity is non-decreasing under measurements~\cite{doi:10.1080/09500349414552171}. Therefore, using Eq.~\eqref{Eq: PhiProjRel}, we find that
     \begin{align}
        &F\left(\mathbb{I}_{A} \otimes \cE^{s}_{D\to A'}(\psi_{AD}),\proj{\Phi_{AA'}}\right)  \\ 
        & \leq  F\left( \sum_{x,y} \proj{xy}\left(\mathbb{I}_{A} \otimes \cE^{s}_{D\to A'}(\psi_{AD})\right)\proj{xy},\Pi_{XX} \left(\proj{\Phi_{AA'}}\right)\right)  \\
        & =  F\left( \Pi_{XX}\left(\mathbb{I}_{A} \otimes \cE^{s}_{D\to A'}(\psi_{AD})\right),\Pi_{XX} \left(\proj{\Phi_{AA'}}\right)\right)  \\
         & =  \frac{1}{2^{2n}}\left( \sum_{x} \sqrt{\bra{x} \hat{\cE}^{s}\left(\proj{f(x)}\right) \ket{x}} \right)^2\\
        & \leq  \frac{1}{2^{n}} \sum_{x} \bra{x} \hat{\cE}^{s}\left(\proj{f(x)}\right) \ket{x}\\
        & = \Pr_{x \leftarrow \bits^n}[\hat{\cE}^{s}\left(\proj{f(x)}\right) = x] \;.
     \end{align}
     The third line is a consequence of the fact that one may ignore contributions for which $x\neq y$ since measuring $\proj{\Phi_{AA'}}$ in the computational basis will never produce such outcomes and the second-to-last line follows from the concavity of the square root.   }
    We know that for any channel using at most $s$ gates  
    \begin{equation*}
        \Pr_{x \leftarrow \bits^n}[\hat{\cE}^{s}(f(x)) = x] \le 2^{-n/2} \;.
    \end{equation*}
    Therefore   $F(\mathbb{I}_{A} \otimes \cE^{s}_{D\to A'}(\psi_{AD}),\proj{\Phi_{AA'}}) \leq 2^{-n/2}$ 
    we get:
    \begin{equation*}
        \hmincompc[s](A|D)_{\rho} = -\log \max_{\cE^{s}_{D\to A'}} d_{A} F(\mathbb{I}_{A} \otimes \cE^{s}_{D\to A'}(\psi_{AD}),\proj{\Phi_{AA'}}) \ge - -n/2 \;. \qedhere
    \end{equation*}
\end{proof}

Given the above, in order to define a dual entropy to the computational min-entropy, one has to choose a fixed purification to work with. We choose the  ``pretty good purification''~\cite{Winter_2004prettygoodpure}. This specific choice of purification may, in general, not always be efficiently preparable, but having access to it does not straightforwardly provide quantum circuits $\cE^s_{C\to A'}$ with any undue advantage (or disadvantage). In principle, however, alternative choices for the purification can be suitable as well.

\begin{definition}[Pretty Good Purification]\label{Def:PrettyGoodPurification}
    For a state $\rho_{A}\in \cS_{\bullet}(A)$ we define the pretty good purification as
    \begin{equation*}
        \proj{\rho^{\textrm{pg}}}_{AC} \coloneq \diff{d_{A}}
        \left(\sqrt{\rho_{A}} \otimes \mathbb{I}_C\right)
        \proj{\Phi}_{A|C}
        \left(\sqrt{\rho_{A}} \otimes \mathbb{I}_C\right)\;,
    \end{equation*}
    where $\ket{\Phi}_{A|C}$ is the maximally entangled state between $A$ and $C$ in the computational basis.
\end{definition}

\begin{remark}\label{Remark:Purification_Is_Not_A_Channel}
    Note that the mapping $\rho_{A}\to \proj{\rho^{\textrm{pg}}}_{AC}$ does not correspond to a valid quantum channel, and is in fact not linear.
    Purifications should be seen as a mathematical tool. Any mixed state can be written as a marginal of a pure state on a larger Hilbert space that we think of as the environment. The choice of a specific purification can be viewed as an assumption on the structure of the state on this environment register.
\end{remark}

\begin{definition} [Quantum Computational Max-Entropy] \label{Def:DualRel}
    Let $\rho_{AB}\in \cS_{\bullet}(AB)$ be a bipartite state and let $\proj{\rho^{\textrm{pg}}}_{ABC}$ be the pretty good purification for $\rho_{AB}$ given in Definition~\ref{Def:PrettyGoodPurification}.
    The quantum computational max-entropy is  defined as
    \begin{equation} \label{Eq:CompDualRel}
        \hmaxcompc[s](A|B)_\rho \coloneq - \hmincompc[s](A|C)_{\rho^{\textrm{pg}}}\;.
    \end{equation}
\end{definition}
We now derive an alternative expression for the computational max-entropy from Definition~\ref{Def:DualRel}.
\begin{lemma}[Alternative Expression]\label{Prop:nosmooth_quantum_comp_maxentropy}
    For any bipartite quantum state $\rho_{AB}\in \cS_{\bullet}(AB)$, and $s \in \mathbb{N}$, it holds that
    \begin{equation*}
        \hmaxcompc[s](A|B)_{\rho} =  \log \max_{\sigma_{B}} \max_{\cE\in \cC(s)} d_{A} F\left((\mathbb{I}_{AB} \otimes \cE_{C\to A'} )   \proj{\rho^{\textrm{pg}}}_{ABC},\ketbra{\Phi_{AA'}}{\Phi_{AA'}}\otimes \sigma_{B} \right) \;,
    \end{equation*}
    where $\ket{\rho^{\textrm{pg}}}_{ABC} = 
        \diff{\sqrt{d_{AB}}}\left(\sqrt{\rho_{AB}} \otimes \mathbb{I}_C\right)
        \ket{\Phi}_{AB|C}$, and one is optimizing over all mixed quantum states $\sigma_{B}$ and quantum circuits $\cE^s_{C\to A'}$ acting on $C$, of size at most $s$.
\end{lemma}

\begin{proof}
    By definition,
    \begin{equation*}
        \hmaxcompc[s](A|B)_{\rho} = \log d_{A} \max_{\cE^s \in \cC(s)}  F\left( (\mathbb{I}_A \otimes \cE^s)(\rho_{AC}), \proj{\Phi_{AA'}} \right) \;,
    \end{equation*}
    where $\rho_{AC} = \operatorname{Tr}_{B} \left[\ketbra{\psi}{\psi}_{ABC} \right]$ and $ \ket{\psi}_{ABC} = (\sqrt{\rho_{AB}} \otimes \mathbb{I}_C)\ket{\Phi_{AB|C}}$. Moreover, due to~\cite[Corollary 3.1]{Tomamichel_2016}, it must hold that
    \begin{align*}
        \hmaxcompc[s](A|B)_{\rho}
         & = \log d_{A} \max_{\cE^s \in \cC(s)}  F\left( (\mathbb{I}_A \otimes \cE^s)(\rho_{AC}), \proj{\Phi_{AA'}} \right) \\
         & =\log d_{A} \max_{\substack{
        \sigma_{A A^\prime B} \ \text{s.t.}                                                                                 \\
                \ \operatorname{Tr}_{B}\left[   \sigma\right]=\proj{\Phi}}
        }
        \max_{\cE^s\in \cC(s)}  F\left((\mathbb{I}_{AB} \otimes \cE_{C\to A'} )   \proj{\rho^{\textrm{pg}}}_{ABC},\sigma_{A A^\prime B} \right) \;.
    \end{align*}
    Here, $\sigma_{AA^\prime B}$ is simply an extension  of $\ket{\Phi_{AA'}}$. However, since $\ket{\Phi_{AA'}}$ is pure, all extensions must be of the form $\sigma_{AA^\prime B} = \proj{\Phi_{AA'}} \otimes \sigma_{B}$. It therefore follows that
    \begin{equation*}
        \hmaxcompc[s](A|B)_{\rho} =  \log \max_{\sigma_{B}} \max_{\cE^s\in \cC(s)} d_{A} F\left((\mathbb{I}_{AB} \otimes \cE_{C\to A'} )  \proj{\rho^{\textrm{pg}}}_{ABC},\ketbra{\Phi_{AA'}}{\Phi_{AA'}}\otimes \sigma_{B} \right) \;. \qedhere
    \end{equation*}
\end{proof}

This alternative expression has an operational meaning as a quantum task. The state $ \left(\sqrt{\rho_{AB}} \otimes \mathbb{I}_{C} \right)\ket{\Phi_{AB|C}}$ is a purification of $\rho_{AB}$, and it contains entanglement between the cut $AB$ and $C$. The set of circuits we are optimizing over aims to distill the entanglement between the registers $A$ and $C$, and in the process, the entangled state is decoupled from the state on  $B$. From this point of view, $\sigma_{B}$ should be viewed as a trash register. The existence of $\sigma_{B}$ used in~\cref{Prop:nosmooth_quantum_comp_maxentropy} is guaranteed by Uhlmann's Theorem~\cite{uhlmann1976transition}. Note that $\sigma_{B}$ may be hard to generate~\cite{bostanci2023unitary}. The focus of this definition is on the efficiency of the entanglement distillation or decoupling protocol~$\cE^s_{C\to A'}$, and not the complexity of the closest product state $\Phi_{AA'}\otimes \sigma_{B}$.

We illustrate the process that we use in~\cref{Prop:nosmooth_quantum_comp_maxentropy} to define the computational max entropy in~\cref{fig:Es-circuit-max-entropy}. We emphasize that the purification part of this process does not correspond to a quantum channel or a physical process that can be performed on the marginal $\rho_{AB}$.

\begin{figure}[ht]
    \centering
    \begin{quantikz}[row sep=0.3cm]
        \lstick[wires=8]{$\ket{\rho^{\textrm{pg}}_{AB|C}}$} & \qw & \qw \rstick[wires=2]{$A$} \\
        & \qw & \qw                       \\
        & \qw & \qw \rstick[wires=2]{$B$} \\
        & \qw & \qw                       \\
        & \gate[wires=4]{\mathcal{E}^s_{C \to A'}} & \qw \rstick[wires=4]{$A'$} \\
        &                                   & \qw                        \\
        &                                   & \qw                        \\
        &                                   & \qw
    \end{quantikz}
    \quad $\approx$ \quad $\proj{\Phi_{AA'}}\otimes \sigma_B$
    \caption{A circuit $\mathcal{E}^s_{C \to A'}$ acts on the purifying system $C$ for the pretty good purification of $\rho_{AB}$. The goal is to distill entanglement between $A$ and $C$ and decoupling from $B$, quantified by fidelity with $\proj{\Phi_{AA'}} \otimes \sigma_B$.}
    \Description[Entanglement distillation circuit]{Quantum circuit with eight input wires representing a purification of $\rho_{AB}$ conditioned on $C$. 
  A four-wire operation $\mathcal{E}^s$ acts on the $C$ subsystem and outputs four wires labeled $A'$. 
  The target state is close to a maximally entangled pair between $A$ and $A'$ tensor an arbitrary state on $B$.}

    \label{fig:Es-circuit-max-entropy}
\end{figure}
\subsubsection{Non-Conditional Entropy}
Whenever the register $B$ is trivial, it follows directly from the alternative expression of the max-entropy that the expression simplifies to
    \begin{equation*}
        \hmaxcompc[s](A)_{\rho} =  \log \max_{\cE\in \cC(s)} d_{A} F\left((\sqrt{\rho_{A}} \otimes \cE_{C\to A'} )\ket{\Phi_{A|C}},\ket{\Phi_{AA'}} \right) \;.
    \end{equation*}

\begin{definition}[Smooth Computational Max-Entropy]\label{def:smooth_max_comp_entropy}
    For any state $\rho_{AB}\in \cS_{\bullet}(AB)$, ${s\in\mathbb{N}}$, and $\varepsilon\ge 0$ the smooth computational max-entropy as given by:
    \begin{equation*}
        \hmaxcompc(A|B)_{\rho} \coloneq  \min_{\tilde{\rho}\in \cB_{\varepsilon}(\rho_{AB})} \hmaxcompc[s](A|B)_{\tilde{\rho}} 
    \end{equation*}
\end{definition}

\begin{restatable}{open_question}{SmoothDualityRelation}\label{OQ: SmoothDualityRelation}
It is unclear whether the smoothed quantities satisfy an exact duality relation. Can it be shown, however, that for a pure state $\rho_{ABC}\in \cS_{\bullet}(ABC)$
\begin{align*}
    \hmaxcompc(A|B)_{\rho} + \hmincompc(A|C)_{\rho} \approx 0 \; ?
\end{align*}
\end{restatable}

\subsubsection{Properties of the Computational Max-Entropy}\label{subse:properties_of_comp_max}

The quantum computational max-entropy satisfies properties that are analogous to the quantum computational min-entropy. The following lemma provides a useful overview of them.
\begin{restatable}{theorem}{PropertiesCompMaxEnt} \label{lem:PropertiesCompMaxEnt}
    The smooth computational max-entropy, $\hmaxcompc(A|B)_{\rho}$, satisfies
    \begin{enumerate}
        \item (Monotonicity in $s$): For any $\rho_{AB}\in \cS_{\bullet}(AB)$, $\varepsilon \geq 0$, and $s^\prime \geq s$,
              \begin{align}
                  \hmaxcompc(A|B)_{\rho} \le \hmaxcompc[\left(s',\varepsilon\right)](A|B)_{\rho}  \;.
              \end{align}
        \item (Monotonicity in $\varepsilon$): For any $\rho_{AB}\in \cS_{\bullet}(AB)$, $s \in \mathbb{N}$, and $\varepsilon^\prime \geq \varepsilon \geq 0$,
              \begin{align}
                  \hmaxcompc[\left(s,\varepsilon^\prime\right)](A|B)_{\rho} \le \hmaxcompc[\left(s,\varepsilon\right)](A|B)_{\rho}  \;.
              \end{align}
        \item ($\mathrm{H}_{\max}^{\varepsilon}$ Relation):  For any $\rho_{AB}\in \cS_{\bullet}(AB)$, $s \in \mathbb{N}$, and $\varepsilon \geq 0$,
              \begin{align}
                  \hmaxcompc(A|B)_{\rho}                   & \le \mathrm{H}_{\max}^{\varepsilon}(A|B)_{\rho}\;, \label{Eq:CompmaxHmaxEpsRelation} \\
                  \lim_{s\to\infty} \hmaxcompc(A|B)_{\rho} & = \mathrm{H}_{\max}^{\varepsilon}(A|B)_{\rho}  \;.
              \end{align}
    \end{enumerate}
\end{restatable}
For a proof of Lemma~\ref{lem:PropertiesCompMaxEnt}, we refer the reader to Appendix~\ref{App:Proofs}. 

As the reader may directly notice, some core properties that one would expect the entropy to have, e.g., data processing, are not included in Lemma~\ref{lem:PropertiesCompMaxEnt}.
We phrase the main ``missing'' properties as an open question below and leave it for future work.

\begin{restatable}{open_question}{MissingPropertiesCompMaxEnt}\label{Prop:MissingPropertiesCompMaxEnt}
    For a quantum state $\rho_{AB}\in \cS_{\bullet}(AB)$ and $s\in\mathbb{N},\varepsilon\ge 0 $ are there $s'\in \mathbb{N},\varepsilon'\ge 0$ such that the following properties hold
    \begin{enumerate}
        \item (Data Processing):  Let $\cC_{B\to B'}$ be a quantum channel that can be implemented using a circuit of size $t$. Then, for any $\rho_{AB}\in \cS_{\bullet}(AB)$, $s \in \mathbb{N}$, and $\varepsilon \geq 0$,
              \begin{align}
                  \hmaxcompc[\left(s',\varepsilon'\right)](A|B')_{\cC(\rho)} \ge \hmaxcompc(A|B)_{\rho} \;.
              \end{align}
        \item (Leakage Chain Rule):  For any $\rho_{ABC}\in \cS_{\bullet}(ABC)$, such that $\ell = \log\dim(C)$,
              \begin{align}
                  \hmaxcompc(A|BC)_{\rho} \ge \hmaxcompc[(s',\varepsilon')](A|B)_{\rho} - 2\ell \;.
              \end{align}
        \item (Smoothing cq-States):   For any cq-state $\rho_{XE}\in \cS_{\bullet}(XE)$ and $\ell = \log\dim(X) \in \mathbb{N}$, there exists a cq-state $\tilde{\rho}\in \cB^{\mathrm{cq}}_{\varepsilon}(\rho)$ such that
              \begin{align}
                  \hmaxcompc(X|E)_{\rho} \le \hmaxcompc[s](X|E)_{\tilde{\rho}} \le \hmaxcompc[\left(s^\prime,\varepsilon\right)](X|E)_{\rho} \;.
              \end{align}
    \end{enumerate}
\end{restatable}

A key technical obstacle in proving the above properties for the computational max-entropy is that its definition, Definition~\ref{def:smooth_max_comp_entropy}, includes a purification of the state. A purification, however, is not a quantum channel and, so, working with it directly is somewhat problematic. At the time of writing, we are unsure whether the aforementioned properties hold. Regardless, proving or disproving them will provide valuable insights for advancing computational entropies.

\section{Applications and Open Questions}

\subsection{Computational Min-Entropy and Pseudo-Randomness} \label{Sec: Computational Min-Entropy and Pseudo-Randomness}
One key benefit of our definition of the computational min-entropy is that it allows us to quantify the uncertainty associated with computational hardness. As we see by Lemma~\ref{Lem:AltPurMinEnt}, there are states that have very low information-theoretical min-entropy, even $0$, but have high computational min-entropy. This means that there are states that have no ``real'' information-theoretic randomness but, nevertheless, cannot be predicted with high probability by a computationally bounded adversary.

The ideal randomness resource for cryptography is a random variable that is uniformly distributed and independent of any side information an adversary may hold. In practice, however, sources of randomness are rarely perfect and may be partially predictable or correlated with adversarial knowledge.

A key tool in cryptography is privacy amplification, the process of producing nearly uniform and independent randomness out of imperfect randomness. The leftover hashing lemma~\cite{HILL1999pseudorandom,bennett2002generalizedprivacyamp_leftoverhashlema} relates the optimal rate of privacy amplification with the min-entropy of the imperfect sources.
Privacy amplification is typically done using functions known as randomness extractors. Randomness extractors are defined as follows:
\begin{definition}
    A function $\Ext: \bits^{n} \times \bits^{d} \to \bits^{m}$ is a quantum proof $(\varepsilon_{\ext},k_{\ext})$ seed extractor with uniform seed, $Y$, if, for all ccq-states $\rho_{XYE} = \rho_{XE}\otimes \omega_{Y}\in \cS_{\circ}(XYE)$, such that $H_{\min}(X|E) \ge k_{\ext}$:
    \begin{equation*}
    \dist(\rho_{\Ext(X,Y)YE},\omega_{m}\otimes\rho_{YE}) \le \varepsilon_{\ext}\;,
    \end{equation*}
    where $\omega_{m}$ is maximally mixed state over $m$ bits.
\end{definition}
A particularly clean example of an extractor is the inner-product function $\mathrm{IP}(X,Y)$, introduced in~\cite{CG85_inner_product_strong_extractor}. When $X$ is a source with sufficient min-entropy conditioned on $E$, and $Y$ is a uniformly random seed, $\mathrm{IP}(X,Y)$ outputs a bit that is statistically close to uniform and independent of both $E$ and $Y$. This result is known to hold even when the side information $E$ may be quantum~\cite{CDNT98quantumentanglementcommunicationcomplexityinnerproduct}.

Randomness extractors enable the generation of useful randomness, that is, provably close to uniform and independent from side information, from sources with a weaker form of randomness characterized by min-entropy.
A natural question to ask is whether this can be extended to computational min-entropy. That is, can sources with sufficiently high computational min-entropy be used to generate better randomness? We cannot hope to extract true, information-theoretical randomness out of sources with high computational min-entropy. However, one can extract something that looks like uniform and independent randomness to a computationally bounded observer. To formalize this notion, let us recall the definition of computational distance from Section~\ref{sec:pre_cm}:

\defcompdist*

We say that states that are close to uniform and independent in the computational distance are \emph{pseudo-random}. In the recent work~\cite{arnon2025computational}, it was shown that pseudo-randomness can be extracted from cq-states with sufficiently high unpredictability entropy, using the inner product extractor as a single-bit seeded extractor. Formally:
\begin{theorem}[\cite{arnon2025computational}, Theorem 2]\label{thm:IPext}
    Let $\rho_{XE}\in \cS_{\circ}(XE)$ be a cq-state where $X$ is distributed over $\bits^{n}$ and $Y$ is uniformly distributed over~$\bits^{n}$. Let $k_{\ext} \in \mathbb{N}, \varepsilon_{\ext}>0$ and ${k_{\ext} \ge 1-2\log(\varepsilon_{\ext})}$. We denote by~$\IP(X,Y)$ the binary inner-product of the values taken by $X$ and $Y$.
    If
    \begin{equation*}
        \hunps[2s+2n+5]^{\varepsilon}(X|E) \ge k_{\ext} \;,
    \end{equation*}
    then
    \begin{equation*}
        \dist_{s}(\rho_{\IP(X,Y)YE},\omega_{1}\otimes\rho_{YE}) \le \varepsilon_{\ext} + 2\varepsilon \;.
    \end{equation*}
\end{theorem}
This demonstrates that the relation between randomness extractors and min-entropy extends to the computational setting. The inner-product extractor only outputs one bit and requires a seed that is as long as the min-entropy source. In~\cite[Lemma 11]{arnon2025computational}, parts of the seed are reused to construct extractors that output~$m$ bits using seeds of length at most $2n^2\log(4m)$, with similar parameters against computational adversaries.

The equivalence between the computational min-entropy and unpredictability entropy shown in Lemma~\ref{Lem:PropertiesCompMinEnt} implies that these results can also be stated for computational min-entropy sources with similar parameters.

\begin{restatable}{lemma}{IPforcompmin}\label{lem:IPext_compmin}
    Let $\rho_{XE}\in \cS_{\circ}(XE)$ be a cq-state where $X$ is distributed over $\bits^{n}$ and $Y$ be uniformly distributed over~$\bits^{n}$. Let $k_{\ext} \in \mathbb{N}, \varepsilon_{\ext}>0$ and $k_{\ext} \ge 1-2\log(\varepsilon_{\ext})$. We denote by~$\IP(X,Y)$ the binary inner-product of the values taken by $X$ and $Y$.
    If
    \begin{equation*}
        \hmincompc[(2s+3n+5,\varepsilon)](X|E) \ge k_{\ext} \;,
    \end{equation*}
    then
    \begin{equation*}
        \dist_{s}(\rho_{\IP(X,Y)YE},\omega_{1}\otimes\rho_{YE}) \le \varepsilon_{\ext} + 2\varepsilon \;.
    \end{equation*}
\end{restatable}

\begin{proof}
    For cq-states, we know from the equivalence between the computational min-entropy and unpredictability entropy shown in~\cref{Lem:PropertiesCompMinEnt} that
    \begin{equation*}
        \hunps[s]^{\varepsilon}(X|E)_{\rho} \ge \hmincompc[(s+n,\varepsilon)](X|E)_{\rho} \;.
    \end{equation*}
    Therefore if $\hmincompc[(2s+3n+5,\varepsilon)](X|E)_{\rho}\ge k_{\ext}$ we get $\hunps[2s+2n+5]^{\varepsilon}(X|E) \ge k_{\ext}$ and by~\cref{thm:IPext}
    \begin{equation*}
        \dist_{s}(\rho_{\IP(X,Y)YE},\omega_{1}\otimes\rho_{YE}) \le \varepsilon_{\ext} + 2\varepsilon \;. \qedhere
    \end{equation*}
\end{proof}

The discussion above establishes a concrete connection between computational min-entropy and pseudo-randomness extraction in the quantum setting, specifically through the use of the inner-product function as a seeded extractor. This illustrates how computational min-entropy can be harnessed to produce pseudo-random outputs. Nevertheless, much remains unknown regarding the broader landscape of extractors in this computational setting.

A central open direction is to characterize the families of extractor functions that can extract pseudo-randomness from sources with only computational min-entropy, as opposed to the much more restrictive information-theoretic condition. While the inner-product function provides one explicit example, it is unclear to what extent this property is shared by other classical extractor families. 

\begin{restatable}{open_question}{ExtractorFamilies}\label{OQ:Comp_ext_pseudorandomness}
    Which families of randomness extractors can be used to (efficiently) extract pseudo-randomness from sources with sufficiently high computational min-entropy?
\end{restatable}

Possible directions for tackling the above open question are discussed in our previous work~\cite{arnon2025computational}.
\subsection{Computational Min-Entropy and Decoupling}
Apart from applications involving pseudo-randomness extractors, one should also explore the possibility of relating the computational min-entropy to computational decoupling. That is, if a quantum state $\rho_{AE}$ has high computational min-entropy of $A$ conditioned on $E$, does this imply that $A$ can be computationally decoupled from~$E$? 

In many ways, decoupling can be viewed as akin to randomness extraction in a fully quantum setting. For decoupling, the objective is to act on $A$, which may be partially entangled or correlated with an environment $E$, such that after applying a suitable quantum operation, the resulting state on $A'$ is nearly uniform and decoupled from the environment $E$, the seed $Y$, and any measurement outcome data, $X$, that is generated during decoupling~\cite{Dupuis_2014_one_shot_decup}. In other words, using a uniformly random seed, one can view it as a mapping
\begin{eqnarray}
\mathcal{D}: \cS_{\circ}({AY})\ &\to& \cS_{\circ}({A^\prime X Y}) \; ,
\end{eqnarray}
which satisfies
    \begin{eqnarray} \label{Eq: InfoTheoretDecoupling}
\frac{1}{2} \Vert \mathcal{D} \left(\rho_{AE} \otimes \omega_{Y} \right)- \omega_{A^\prime}  \otimes \rho_{XYE}\Vert_{1} \leq \varepsilon \; .
\end{eqnarray}
One approach towards decoupling $A$ from the environment is discussed in~\cite{Dupuis_2014_one_shot_decup}, where $\mathcal{D}$ is constructed as follows.\footnote{In fact, the results from~\cite{Dupuis_2014_one_shot_decup} are slightly stronger than Eq.~\eqref{Eq: InfoTheoretDecoupling}, as all registers are decoupled, i.e. ${\rho_{XYE} \approx \omega_{X} \otimes \omega_{Y} \otimes \rho_{E}}$.} The seed $Y$ randomly samples a unitary from an (approximate) two-design. After applying this unitary on $A$, a subset of the qubit registers on $A$ are measured in the computational basis.\footnote{In principle, the work~\cite{Dupuis_2014_one_shot_decup} also discusses how decoupling is possible if one additionally traces out some subsystems of $A$. However, for this brief overview, we do not focus on this contribution.} Using this approach, the optimal, i.e. maximal, number of qubits that \textit{do not} need to be measured for Eq.~\eqref{Eq: InfoTheoretDecoupling} to hold is characterized by the smooth min-entropy of $A$ conditioned on $E$ in the quantum information theoretical setting~\cite{Dupuis_2014_one_shot_decup}.

It is natural to ask if, like extractors, decoupling can also be extended to the computational setting. Here, the LHS of Eq.~\eqref{Eq: InfoTheoretDecoupling} would be replaced with the computational distance between $\mathcal{D} \left(\rho_{AE} \otimes \omega_{Y} \right)$ and $ \omega_{A^\prime}  \otimes \rho_{XYE}$, and one would expect the dimension of the register $A^\prime$ to be related to the computational min-entropy. 
Understanding whether this task is possible for any quantum state with sufficiently high computational min-entropy is a fundamental question. Natural candidates for achieving computational decoupling include (approximate) unitary two-designs~\cite{Szehr_2013}, which can be constructed efficiently, see e.g.~\cite{Nakata_2017, DiVincenzo_2002, Harrow_2009}, but their power in the computational setting remains open.

\begin{restatable}{open_question}{ComputationalDecoupling}\label{OQ:ComputationalDecoupling}
    Is it possible to (efficiently) computationally decouple any quantum state from its environment, given that it has sufficiently high computational min-entropy?
\end{restatable}

Note that the efficiency of unitary two-designs is not enough to resolve this open question. Rather, it has to be explicitly shown that if the quantum output is not computationally close to uniform and decoupled from other systems, it must hold that the input state did not have high computational min-entropy.

\subsection{Computational Max-Entropy and Entanglement Distillation/Data Compression}

In this work, we define the computational max-entropy via a duality relation, mirroring the information-theoretic case and ensuring that it inherits a range of structural properties. Beyond this formal definition, it is also directly associated with an operational task: for a fixed circuit-size bound $s$, it quantifies the maximal fidelity achievable by an $s$-bounded channel acting on the auxiliary register $C$ of the pretty good purification, thereby distilling entanglement between $A$ and $C$ while decoupling $A$ from $B$. At the same time, this definition does not preserve all of the operational meanings of the information-theoretic max-entropy.

To see why this is the case, let us first consider the information-theoretic task of entanglement distillation between two parties. Given any bipartite state $\rho_{AB}$, it is known that approximately $-\mathrm{H}_{\max}^{\varepsilon}(A|B)_{\rho}$ Bell pairs can be distilled via the decoupling protocol~\cite{khatri2024principlesquantumcommunicationtheory}. This protocol takes advantage of isometries, which are generally not considered to be efficiently implementable~\cite{bostanci2023unitary}.\footnote{The complexity framework used in~\cite{bostanci2023unitary} when discussing the closely related Channel Decoding Problem is of a slightly different flavor that ours. However, as is discussed in~\cite{arnonfriedman2023computationalentanglementtheory}, entanglement distillation is considered to generally be a computationally inefficient task, and we expect that the results from~\cite{bostanci2023unitary} can be related closer to our formalism.} The reason behind this high computational cost can be related to the Uhlmann transformation problem~\cite{bostanci2023unitary}. More concretely, the computational complexity arises from the following observations: The decoupling protocol generates classical (measurement output) data that acts as an error syndrome.\footnote{If one follows the decoupling protocol from~\cite{khatri2024principlesquantumcommunicationtheory}, then the protocol only explicitly implements a single measurement, which occurs on $A$ and not on $B$. However, this measurement also generally collapses a subsystem of $B$. Although encoded in a larger quantum register, this subsystem essentially contains classical data, and should be considered part of the syndrome.} That is, the value of the classical data contains the information as to which (local) quantum circuit needs to be applied so that the output state is maximally entangled. However, whenever there are exponentially many different error syndromes that need to be accounted for,\footnote{As is the case, for example, when one starts with IID noisy initial copies.} coding all of this information into an efficient circuit may generally not be possible under standard complexity assumptions. This is, for example, the case with quantum error correcting codes, where it is known that decoding them is generally a $\#\mathrm{P}$-Complete problem, see e.g.~\cite{iyer2013hardness}. Apart from the potential complexity of associating to each error syndrome the correct decoding circuit, even when there are solely two types of errors it may be that distinguishing between them is a computationally expensive task. In~\cite[Section 9.1.1]{bostanci2023unitary}, an example to this effect is provided, using EFI pairs~\cite{brakerski2022computational}.

Ideally, one would hope that $- \hmaxcompc(A|B)_{\rho}$ analogously provides an approximate lower bound on the efficiently distillable entanglement. However, not only is $- \hmaxcompc(A|B)_{\rho}$ decreasing in $s$, but due to Eq.~\eqref{Eq:CompmaxHmaxEpsRelation}, it holds that
\begin{align*}
    - \hmaxcompc(A|B)_{\rho}  \geq -\mathrm{H}_{\max}^{\varepsilon}(A|B)_{\rho} \; .
\end{align*}
This would indicate that there always exist efficient distillation protocols that can rival, and even surpass, the decoupling protocol. Furthermore, due to the scaling of the computational max-entropy as a function of $s$, this would indicate that the circuit-depth of distillation protocols is negatively correlated to the distillation rate. Given these observations, it seems unreasonable for the computational max-entropy to be related to the usual notion of entanglement distillation.

Similarly, on classical states, the information-theoretic max-entropy is related to the task of data compression~\cite{KRS09OperationalMeaningEntropy, bostanci2023unitary}. For a classical, single-register state $\rho_{X}$, the minimal encoding length is approximately given by $\mathrm{H}_{\max}^{\varepsilon}(X)_{\rho}$, see e.g.~\cite{KRS09OperationalMeaningEntropy}. Again, we find that $ \hmaxcompc(X)_{\rho}$ is increasing in $s$ and
\begin{align*}
    \hmaxcompc(X)_{\rho}  \leq \mathrm{H}_{\max}^{\varepsilon}(X)_{\rho} \; .
\end{align*}
If the encoding length for efficient data compression were given by the computational max-entropy, then this would both imply that efficient circuits can essentially compress data at an optimal rate, and that decreasing the circuit length has no negative impact on the compression rate. However, as is discussed in~\cite[Section 9.3]{ bostanci2023unitary}, if all efficiently generateable distributions can be optimally compressed, it would follow from~\cite{yaotrapdor82, HILL1999pseudorandom} that one-way functions do not exist. Moreover, whenever there is a sufficiently large gap between the computational max-entropy and the information-theoretic max-entropy, then this would indicate that efficient circuits can compress data more efficiently than is information-theoretically possible.

Given this discussion, we pose the following open question.
\begin{restatable}{open_question}{MaxEntropyOperationalMeaning}\label{OQ: MaxEntropyOperationalMeaning}
    What relevant operational tasks can be related to the computational max-entropy, $ \hmaxcompc(A|B)_{\rho}$?
\end{restatable}
Rather than using duality relations, one can define variants of the computational max-entropy with these tasks in mind. However, they will, in general, not satisfy duality relations, nor is it immediately clear how one can establish a single definition that relates to both efficient entanglement distillation and efficient data compression. 
\begin{restatable}{open_question}{MaxEntropyAlternativeDefinition}\label{OQ: MaxEntropyAlternativeDefinition}
    Is there an alternative, computationally meaningful definition of the max-entropy that accurately captures some quantity of interest relevant to tasks such as efficient entanglement distillation or efficient data compression?
\end{restatable}

One attempt at such a definition is the computational relative entropy defined using hypothesis testing from~\cite{Munson_2025,Yunger_Halpern_2022}. It accurately captures the optimal rate for data compression under computational constraints (though without giving an explicit efficient protocol). This can be potentially used to define a different variant of a computational max-entropy. An interesting avenue for future work would be to relate this quantity to efficient entanglement distillation.

\subsection{Relation to Computational Entanglement Theory, Quantum Cryptography and More}\label{sec:othe_connections}

Entropic quantities play a major role in quantum information theory. They are of particular importance when used to upper or lower bound the rates in which certain tasks can be performed, e.g., key rates in quantum key distribution or the rate in which maximally entangled states can be distilled from a given state. The relations between the rates and the entropies give the latter their operational meaning.

Recent studies on pseudo-entanglement~\cite{gheorghiu2020estimating,aaronson2023quantumpseudoentanglement}, for example, revealed that the \emph{information-theoretic} entropies do not carry their operational meaning once we limit the computational power of the protocols~\cite{arnonfriedman2023computationalentanglementtheory,leone2025entanglementtheorylimitedcomputational}.
It is thus natural to wonder whether the newly defined computational entropies (or other variants of them) are now the quantities that should be used to prove lower and upper bounds on operational tasks with limited computational power. For example, is it possible to relate specific computational entropies to the computational distillable entanglement and computational entanglement cost defined in~\cite{arnonfriedman2023computationalentanglementtheory}?

In a different direction,  information-theoretic entropies are a main tool in quantum cryptography with information-theoretic security. This leads to an interesting fundamental question-- can the computational entropies act as a useful tool for the analysis of cryptographic protocols in which the quantum adversary has a limited computational power? The results of~\cite{arnon2025computational} indicate that for pseudo-randomness extractors this is indeed possible. It is thus intriguing to consider further applications in cryptography, classical and quantum~\cite{alagic2016computational,barooti2023public,bernstein2017post}.

Lastly, entropic quantities arise in various ways in the study of the AdS/CFT correspondence~\cite{ryu2006holographic,ryu2006aspects,hayden2013holographic,akers2021leading,akers2023one}. It is natural to hypothesize that computational entropies should come about when studying the computational complexity of the correspondence itself. In that context, computational entropies might even be a more fruitful object to examine rather than pseudo-entanglement~\cite{gheorghiu2020estimating,aaronson2023quantumpseudoentanglement,arnonfriedman2023computationalentanglementtheory}.

\section{Open Questions}\label{sec:OQ}
Throughout this work, we raised several open questions regarding properties of the computational entropies. In this section, we collect and summarize these questions in one place to provide a unified view of the challenges that remain and to highlight directions for future research. 

Open Questions~\ref{OQ:Comp_ext_pseudorandomness} is about generalizing the result regarding Trevisan's extractors to other families of extractors. 

\ExtractorFamilies*

Open Questions~\ref{OQ:ComputationalDecoupling} is about generalizing the result from the classical task of pseudo-randomness extraction to the fully quantum task of computational decoupling. 

\ComputationalDecoupling*

While the computational max-entropy is defined by duality, it is not clear how well this duality holds when we allow for slight perturbations, by smoothing using the purified distance. Open Questions~\ref{OQ: SmoothDualityRelation} is about this duality in the smooth case.

\SmoothDualityRelation*

We saw that, when defined by duality, the computational max-entropy preserves some of the properties expected from a computational entropy measure. Several other properties remain as an open question.

\MissingPropertiesCompMaxEnt*

In contrast to the computational min-entropy, which has a clear operational meaning and is related to an operationally meaningful task of extracting pseudo-randomness, the operational meaning of the computational max-entropy remains an open question.

\MaxEntropyOperationalMeaning*

We also refer back to Section~\ref{sec:othe_connections} for further fundamental questions.

\section{Summary}

In this manuscript, we advocate for the development of a computational variant of quantum information theory (QIT).
Our work marks a significant step toward a rigorous computational QIT, establishing foundational computational entropies, alongside a mathematical toolkit and various open directions for future research. 

Motivated by the results of~\cite{arnon2025computational}, we extend the quantum computational unpredictability entropy to arbitrary bipartite states, thereby introducing a fully quantum computational min-entropy. This quantity satisfies several essential properties, including data processing and a leakage chain rule. This quantity lets us directly quantify the uncertainty associated with computational hardness in the fully quantum setting. We show that the computational min-entropy is an operationally relevant quantity for cryptographic tasks. It was shown in a partner work~\cite{arnon2025computational} that Trevisan's extractor can be used to extract computationally secret keys from sources with sufficiently high computational unpredictability entropy. As the computational min-entropy is essentially equivalent to the computational unpredictability entropy for cq-states, the same results must hold here as well.

We expect this computational min-entropy to be a foundational building block, hopefully useful for many other cryptographic and decoupling tasks. With this in mind, we believe that 
Open Questions~\ref{OQ:Comp_ext_pseudorandomness} and~\ref{OQ:ComputationalDecoupling} pose interesting avenues for further research.

We also define a computational max-entropy that is dual to the min-entropy. When defining the computational max-entropy in this way, we show that it cannot retain the exact operational meaning of the information-theoretic max-entropy. This raises two interesting questions. First, what tasks can this computational max-entropy be related to? Secondly, how should one construct an alternate definition for the computational max-entropy whose operational meaning more closely resembles the information-theoretic case? 

Although the computational max-entropy must have a different operational meaning, it does satisfy certain monotonicity properties and converges to the information-theoretic max-entropy in the asymptotic circuit-size limit. Moreover, as stated in Open Questions~\ref{OQ: SmoothDualityRelation} and~\ref{Prop:MissingPropertiesCompMaxEnt}, it is possible that this quantity satisfies other desirable properties. Finding an answer to these questions would provide further insight as to how useful this quantity can be.

Together, the computational min- and max-entropies enrich the emerging toolkit of computational QIT. We expect the min-entropy to serve as a foundational primitive for pseudo-randomness extraction and other quantum information processing tasks, such as decoupling protocols. Moreover, a deeper understanding of the max-entropy may illuminate how computational limitations affect the well-established information-theoretic duality relation between min and max-entropy. Addressing the open questions (also highlighted in the following section for convenience), we believe, will advance the development of a full-fledged computational QIT.

\begin{anonsuppress}
\paragraph{Acknowledgments.}
NA, TAH, and RA were supported by the Peter and Patricia Gruber Award and by the Air Force Office of Scientific Research under award number FA9550-22-1-0391. 
RA was further generously supported by the Koshland Research Fund and is the Daniel E. Koshland Career Development Chair.
JMR was supported by the  CHIST-ERA project “Modern Device Independent Cryptography" and the
ETH Zurich Quantum Center. 

\end{anonsuppress}

\bibliographystyle{unsrtnat}
\bibliography{bib}

\appendix

\section*{Appendix}

    \section{Proofs for computational entropies} \label{App:Proofs}
    \subsection{Proof of Theorem~\ref{Lem:PropertiesCompMinEnt}: properties of the computational min-entropy}

    \PropertiesCompMinEnt*

    \subsubsection{Monotonicity}
    The information-theoretic smooth min entropy is known to be monotonic in $\varepsilon$. This can be seen by the definition of smooth min-entropy and the fact that the purified distance is indeed a distance measure.
    Similarly, our computational min-entropy is monotonic in both the smoothing parameter $\varepsilon$ and the computational parameter $s$. This is a natural consequence of the definition; in both cases, the set for the optimization for a larger parameter contains the sets for the optimization for smaller parameters. As $\varepsilon$ grows, the definition allows for more error, and as $s$ grows, more circuits are available.
    \begin{lemma}[Monotonicity]
        For any $s' \ge s$ and any bipartite state $\rho_{AB}$,
        \begin{equation*}
            \hmincompc(A|B)_{\rho} \ge \hmincompc[\left(s',\varepsilon\right)](A|B)_{\rho}  \;.
        \end{equation*}
        Moreover,  for any $\varepsilon' \ge \varepsilon \ge 0$,
        \begin{equation*}
            \hmincompc[\left(s,\varepsilon'\right)](A|B)_{\rho} \ge \hmincompc(A|B)_{\rho}  \;.
        \end{equation*}
    \end{lemma}

    \begin{proof}
        Given any computational parameter, $s$, the computational min-entropy $\hmincompc(A|B)_{\rho}$ is defined via an optimization over quantum circuits $\cE_{B\to A'}$ of size at most $s$. If we increase $s$ to $s'\ge s$, the optimization would be over all circuits of size $s'$, which includes all circuits of size $s$. Therefore,
        \begin{align*}
            \hmincompc(A|B)_{\rho}
             & = -\log d_{A} \min_{\tilde{\rho}\in \cB_{\varepsilon}(\rho)}\max_{\cE \in \cC(s)}
            F\left( (\mathbb{I}_A \otimes \cE)(\tilde{\rho}_{AB}), \proj{\Phi_{AA'}} \right)               \\
             & \ge  -\log d_{A} \min_{\tilde{\rho}\in \cB_{\varepsilon}(\rho)}\max_{\cE \in \cC(s^\prime)}
            F\left( (\mathbb{I}_A \otimes \cE)(\tilde{\rho}_{AB}), \proj{\Phi_{AA'}} \right)               \\
             & = \hmincompc[\left(s',\varepsilon\right)](A|B)_{\rho} \;.
        \end{align*}

        Given a smoothing parameter, $\varepsilon$, $\hmincompc(A|B)_{\rho}$ is obtained by optimizing over all states $\tilde{\rho}$ such that $\tilde{\rho}\in\cB_{\varepsilon}(\rho)$. Whenever $\varepsilon' \ge \varepsilon$, if $\tilde{\rho}\in\cB_{\varepsilon^\prime}(\rho)$, then it must also hold that $\tilde{\rho}\in\cB_{\varepsilon}(\rho)$. As such,
        \begin{align*}
            \hmincompc[\left(s,\varepsilon'\right)](A|B)_{\rho}
             & = -\log d_{A} \min_{\tilde{\rho}\in \cB_{\varepsilon^\prime}(\rho)}\max_{\cE \in \cC(s)}
            F\left( (\mathbb{I}_A \otimes \cE)(\tilde{\rho}_{AB}), \proj{\Phi_{AA'}} \right)            \\
             & \ge  -\log d_{A} \min_{\tilde{\rho}\in \cB_{\varepsilon}(\rho)}\max_{\cE \in \cC(s)}
            F\left( (\mathbb{I}_A \otimes \cE)(\tilde{\rho}_{AB}), \proj{\Phi_{AA'}} \right)            \\
             & = \hmincompc[\left(s,\varepsilon\right)](A|B)_{\rho} \;. \qedhere
        \end{align*}
    \end{proof}

    \subsubsection{Data Processing}
    The data processing inequality for smooth min-entropy~\cite[Theorem 18]{Tomamichel_2010_Duality}, shows that uncertainty about the system $A$ given side-information $B$ does not decrease under local operation on $B$.
    In the computational case, the computational resources need to be also taken into account, as operations on the side information, that require computational power, may reduce the computational power needed to guess $A$, compared to the initial state. We formalize this idea in the following lemma.
    \begin{lemma}[Data Processing Inequality]\label{lem:DPI_data_processing_inequality_comp_min_entropy}
        Let $\rho_{AB}$ be a bipartite state, $s\in\mathbb{N}$, $\varepsilon\ge0$, and $\cC_{B\to B'}$ be a quantum channel that can be implemented using a circuit of size $t$. Then,
        \begin{equation*}
            \hmincompc(A|B')_{\cC(\rho)} \ge \hmincompc[\left(s+t,\varepsilon\right)](A|B)_{\rho} \;.
        \end{equation*}
    \end{lemma}

    \begin{proof}
        By definition, there exists a state $\tilde{\rho}_{AB}$ with $\bP(\rho_{AB},\tilde{\rho}_{AB})\le \varepsilon$ such that for \emph{any} circuit $\cE$ of size at most $(s + t)$ acting on $B$,
        \begin{equation*}
            \hmincompc[\left(s+t,\varepsilon\right)](A|B)_{\rho}  =
            -\log d_{A} \max_{\cE \in \cC(s+t)}
            F\left( (\mathbb{I}_A \otimes \cE)(\tilde{\rho}_{AB}), \proj{\Phi_{AA'}} \right) \;.
        \end{equation*}
        By the monotonicity of the purified distance under quantum channels,
        \begin{equation} \label{Eq:DPIEpsTrDist}
            \bP\left((\mathbb{I}_A\otimes\cC)(\rho_{AB}), (\mathbb{I}_A\otimes\cC)(\tilde{\rho}_{AB})\right) \le \varepsilon \;.
        \end{equation}
        Let $\cE^{\prime}$ be any circuit of size at most $s$ acting on $B'$. Composing $\cE'$ with the size-$t$ circuit for $\cC$ yields a circuit $\cE$ of size $s + t$ acting on $B$, such that
        \begin{equation*}
            (\mathbb{I}_A\otimes\cE')\left((\mathbb{I}_A\otimes\cC)(\tilde{\rho}_{AB})\right)
            =
            (\mathbb{I}_A\otimes\cE)(\tilde{\rho}_{AB}) \;.
        \end{equation*}
        By our assumption on $\tilde{\rho}_{AB}$, it must hold that
        \begin{align*}
            \hmincompc[\left(s+t,\varepsilon\right)](A|B)_{\rho}
             & = -\log d_{A} \max_{\cE \in \cC(s+t)}  F\left( (\mathbb{I}_A \otimes \cE)(\tilde{\rho}_{AB}), \proj{\Phi_{AA'}} \right)                                                                                          \\
             & \le  -\log d_{A}\max_{\cE^{\prime} \in \cC(s)}  F\left( (\mathbb{I}_A \otimes \cE^{\prime} ) (\cC(\tilde{\rho}_{AB})), \proj{\Phi_{AA'}} \right)                                                                 \\
             & \le  -\log d_{A} \min_{\tilde{\rho}^{\prime}\in \cB_{\varepsilon}(\cC (\rho))}\max_{\cE^{\prime} \in \cC(s)}  F\left( (\mathbb{I}_A \otimes \cE^{\prime})(\tilde{\rho}^{\prime}_{AB}), \proj{\Phi_{AA'}} \right) \\
             & = \hmincompc[\left(s,\varepsilon\right)](A|B)_{\cC(\rho)} \;.
        \end{align*}
        We highlight that the second-to-last line is a direct consequence of Eq.~\eqref{Eq:DPIEpsTrDist}.
    \end{proof}

    \subsubsection{Leakage Chain Rule}\label{app_leakage_chainrule}
    The leakage chain rule, shown in the information-theoretic setting in~\cite{WTHR11Impossibilitygrowingquantumbitcommit}, gives a bound on how much extra information can reduce the uncertainty about a state. It turns out that for the computational min-entropy entropy, the chain rule closely resembles the information-theoretic chain rule.

    \begin{restatable}[Leakage Chain Rule for Computational Min-Entropy]{lemma}{ChainRuleQuantumCompEntropy} \label{lem:chain_rule_quantum_comp_entropy}
        For any state $\rho_{ABC}$, and any $\varepsilon \ge 0$, $s\in\bbN$, $\ell = \log\dim(C)$, and let $t$ be the size of a circuit generating $\omega_{C}$ we have:
        \begin{equation*}
            \hmincompc(A|BC)_{\rho} \ge \hmincompc[(s+t,\varepsilon)](A|B)_{\rho} - 2\ell \;.
        \end{equation*}
    \end{restatable}

    \begin{remark}
        The exact value of $t$ depends on the choice of universal gate set. Assuming the Haddamard and CNOT gates each require 1 gate, we get $t\le 2\ell$, as can be seen by the circuit in~\cref{fig:circuit_maximal_mixed_state}.
    \end{remark}

    \begin{figure}[ht]
        \centering
        \begin{quantikz}[row sep=0.25cm, column sep=0.4cm]
            \lstick{$\ket{0}$} & \gate{H} & \ctrl{3} & \qw      & \qw      & \qw  \rstick[wires=3]{$\omega_{C}$}\\
            \lstick{$\ket{0}$} & \gate{H} & \qw      & \ctrl{3} & \qw      & \qw                                \\
            \lstick{$\ket{0}$} & \gate{H} & \qw      & \qw      & \ctrl{3} & \qw                                \\
            \lstick{$\ket{0}$} & \qw      & \targ{}  & \qw      & \qw      & \qw \rstick[wires=3]{}             \\
            \lstick{$\ket{0}$} & \qw      & \qw      & \targ{}  & \qw      & \qw                   & \trash{}   \\
            \lstick{$\ket{0}$} & \qw      & \qw      & \qw      & \targ{}  & \qw
        \end{quantikz}
        \caption{A circuit generating the maximally mixed state on $3$ qubits using $3$ $H$ gates and $3$ CNOT gates.}
        \Description[A circuit generating the maximally mixed state]{A circuit generating the maximally mixed state on $3$ qubits using $3$ $H$ gates and $3$ CNOT gates.}
        \label{fig:circuit_maximal_mixed_state}
    \end{figure}

    For the proof of the chain rule, the main tool we needed in the cq-case was the following lemma from~\cite{WTHR11Impossibilitygrowingquantumbitcommit}, an immediate consequence of the pinching inequality~\cite{hayashi_optimal_2002}, and which is true not only for cq states, but for any positive operator.

    \begin{lemma}\label{lem:inequality-of-operators-for-comp-leakage-chain-rule}
        For any state $\rho_{A}$ and any extension $\rho_{AB}$, we have:
        \begin{equation*}
            \rho_{AB} \le \dim(B)^{2} (\rho_{A} \otimes \omega_{B}) \;,
        \end{equation*}
        where $\omega_{B}$ is the maximally mixed state on $B$ (recall~\cref{def:max_mix}).
    \end{lemma}

    Additionally, we need the following lemma, stating that fidelity with a pure state is monotonous with respect to the ordering of positive operators.\footnote{We note that Lemma~\ref{lem:fidelity_monoton_with_pure} holds more generally, i.e., even if the second entry is not pure. This can be seen using the SDP formulation for the fidelity.}
    \begin{lemma}\label{lem:fidelity_monoton_with_pure}
        For any states $\rho,\sigma$ such that $\rho \ge \sigma$, and for any pure state $\proj{\psi}$
        \begin{equation*}
            F(\rho,\proj{\psi}) \ge F(\sigma,\proj{\psi})  \;.
        \end{equation*}
    \end{lemma}

    \begin{proof}
        The fidelity with a pure state can be written as:
        \begin{equation*}
            F(\rho,\proj{\psi}) = \bra{\psi}\rho\ket{\psi} \;.
        \end{equation*}
        Therefore for any $\rho \ge \sigma$ we can write
        \begin{equation*}
            F(\rho,\proj{\psi}) =\bra{\psi}\rho\ket{\psi} \ge \bra{\psi}\sigma\ket{\psi} = F(\sigma,\proj{\psi}) \;.  \qedhere
        \end{equation*}
    \end{proof}

    With these tools, we are ready to prove the fully quantum leakage chain rule for the computational min entropy~\cref{lem:chain_rule_quantum_comp_entropy}.
    \begin{proof}
        Let $\cF$ be a circuit that generates $\omega_C$. It has size $t=\bO(\ell)$.
        Let $\tilde \rho_{AB}\in \cB_{\varepsilon}(\rho_{AB})$ be the optimal smoothed state in $\hmincompc[s+t,\varepsilon](A|B)_{\rho}$.
        By~\cref{lem:extension_uhlmann_property_for_purefied_distance} there exists an extension $\tilde \rho_{ABC}$ of $\tilde \rho_{AB}$ such that $\bP(\tilde{\rho}_{ABC},\rho_{ABC})\le\varepsilon$.
        Therefore, by the definition of the smoothed computational min-entropy,
        \begin{equation*}
            2^{-\hmincompc[s,\varepsilon](A|BC)_{\rho}}\leq \max_{\cE\in \cC(s)}d_{A}F((\mathbb{I}_{A}\otimes\cE) (\tilde{\rho}_{ABC}),\proj{\phi_{AA'}})\,.
        \end{equation*}
        By \cref{lem:inequality-of-operators-for-comp-leakage-chain-rule} and~\cref{lem:fidelity_monoton_with_pure} we then have
        \begin{align*}
            2^{-\hmincompc[s,\varepsilon](A|BC)_{\rho}}
             & \leq \max_{\cE\in \cC(s)}d_{A}|C|^2 F((\mathbb{I}_{A}\otimes\cE) (\tilde{\rho}_{AB}\otimes \omega_C),\proj{\phi_{AA'}}) \\
             & \leq \max_{\cE\in \cC(s)}d_{A}|C|^2 F((\mathbb{I}_{A}\otimes\cE\circ\cF) (\tilde{\rho}_{AB}),\proj{\phi_{AA'}})         \\
             & \leq \max_{\cE\in \cC(s+t)}d_{A}|C|^2 F((\mathbb{I}_{A}\otimes\cE) (\tilde{\rho}_{AB}),\proj{\phi_{AA'}})               \\
             & \leq 2^{2\ell}2^{-\hmincompc[s+t,\varepsilon](A|B)_{\rho}}\,. \qedhere
        \end{align*}
    \end{proof}

    \subsubsection{Purification Chain Rule}\label{app_subsec_purification_chain_rule}
    
    In~\cite{Munson_2025} for the complexity entropy, there is a conjectured chain rule, in a different direction from the leakage chain rule. While the leakage chain rule bounds how much knowledge an adversary can get by getting more qubits, this chain rule, which we call the purification chain rule, measures how much additional qubits can purify a random system. We show that for our definition of computational min entropy, this chain rule holds with no loss in parameters.

    \begin{restatable}[Purification Chain Rule for Computational Min-Entropy]{lemma}{PurificationChainRuleQuantumCompMinEntropy} \label{lem:Purification_chain_rule_quantum_comp_entropy}
       Let $\rho_{ABC}$ be a tripartite state and $\varepsilon \ge 0,s\in\mathbb{N}$. Denote $d_{A}$ the dimension of $\rho_{A}$, 
        \begin{equation*}
            \hmincompc(B|C)_{\rho} \le  \hmincompc(AB|C)_{\rho} + \log(d_{A}) \;.
        \end{equation*}
    \end{restatable}

    \begin{proof}
        Let $\cE^{s}_{C\to A'B'}$ be a channel such that
        \begin{equation*}
                    \hmincompc[s](AB|C)_{\rho} = -\log d_{A} d_{B}  F\left( (\mathbb{I}_{AB} \otimes \cE^{s}_{C\to A'B'})(\rho_{ABC}), \proj{\Phi_{AA'BB'}} \right) \;,
        \end{equation*}
        By discarding $A'$, any channel $\cE^{s}_{C\to A'B'}$ that can be implemented using at most $s$ gates, yields a channel: $\hat{\cE}^{s}_{C\to B'}$ that also uses at most $s$ gates:
        $$\hat{\cE}^{s}_{C\to B'} = \mathrm{Tr}_{A'}\circ\cE^{s}_{C\to A'B'}\;.$$ 
        By the definition of the computational min entropy, we know that for any such channel:
        \begin{equation*}
            \hmincompc[s](B|C)_{\rho}  \le -\log  d_{B}  F\left( (\mathbb{I}_{B} \otimes \hat{\cE}^{s}_{C\to B'})(\rho_{BC}), \proj{\Phi_{BB'}} \right) \;,
        \end{equation*}
        Since the fidelity is monotonically increasing under partial trace and $\Phi_{BB'} = \ptr{AA'}{\proj{\Phi_{AA'BB'}}}$ , we get:    
        \begin{align*}
                    \hmincompc[s](AB|C)_{\rho} 
                    & = -\log d_{A} d_{B}  F\left( (\mathbb{I}_{AB} \otimes \cE^{s}_{C\to A'B'})(\rho_{ABC}), \proj{\Phi_{AA'BB'}} \right) \\
                    & \ge -\log d_{A} d_{B}  F\left( (\mathbb{I}_{B} \otimes \hat{\cE}^{s}_{C\to B'})(\rho_{BC}), \proj{\Phi_{BB'}} \right) \\
                    & \ge  \hmincompc[s](B|C)_{\rho} -\log d_{A} \;. 
        \end{align*}
        This solves the lemma in the non-smooth case. 
        For the smooth case, let $\varepsilon\ge 0$, denote $\tilde{\rho}_{BC}\in \cB_{\varepsilon}(\rho_{BC})$ a state such that there is a channel $\tilde{\cE}^{s}_{C\to B'}$
        \begin{equation*}
            \hmincompc(B|C)_{\rho} = \hmincompc[s](B|C)_{\tilde{\rho}} \;.
        \end{equation*}
        Consider an extension $\tilde{\rho}_{ABC}$, and let $\cE^{s}_{C\to A'B'}$ be the optimal $s$ gate channel for $\hmincompc[s](AB|C)_{\tilde{\rho}}$, that is, a channel such that
        \begin{equation*}
            \hmincompc[s](AB|C)_{\tilde{\rho}}  = -\log d_{A} d_{B}  F\left( (\mathbb{I}_{AB} \otimes \cE^{s}_{C\to A'B'})(\tilde{\rho}_{ABC}), \proj{\Phi_{AA'BB'}} \right) \;.
        \end{equation*}
        Let $\hat{\cE}^{s}_{C\to B'} = \mathrm{Tr}_{A'}\circ\cE^s_{C\to A'B'}$, such that:
        \begin{equation*}
            \hmincompc(B|C)_{\rho} \leq -\log d_{B} F\left( (\mathbb{I}_{B} \otimes \hat{\cE}^{s}_{C\to B'})(\tilde{\rho}_{BC}), \proj{\Phi_{BB'}} \right) \;.
        \end{equation*}
        Note that we pick an optimal state for $\hmincompc(B|C)_{\rho}$, $\tilde{\rho}_{BC}$, but a potentially suboptimal channel. The channel is optimal for the extension $\tilde{\rho}_{ABC}$, the reduced channel $\hat{\cE}^{s}_{C\to B'}$ is not nececeraly optimal for $BC$.
        Using that, we can write:
            \begin{align*}
                    \hmincompc[s](AB|C)_{\tilde{\rho}} 
                    & = -\log d_{A} d_{B}  F\left( (\mathbb{I}_{AB} \otimes \cE^{s}_{C\to A'B'})(\tilde{\rho}_{ABC}), \proj{\Phi_{AA'BB'}} \right) \\
                    & \ge -\log d_{A} d_{B}  F\left( (\mathbb{I}_{B} \otimes \hat{\cE}^{s}_{C\to B'})(\tilde{\rho}_{BC}), \proj{\Phi_{BB'}} \right) \\
                    & \ge  \hmincompc(B|C)_{\rho} -\log d_{A} \;. 
        \end{align*}
        If $\tilde{\rho}_{ABC}\in\cB_{\varepsilon}(\rho_{ABC})$, by definition:
        \begin{equation*}
            \hmincompc(AB|C)_{\rho} \ge \hmincompc[s](AB|C)_{\tilde{\rho}} \;.
        \end{equation*}
        For any state $\tilde{\rho}_{BC}\in \cB_{\varepsilon}(\rho_{BC})$ an extension such that $\tilde{\rho}_{ABC}\in\cB_{\varepsilon}(\rho_{ABC})$ is guaranteed by the extension property of the purified distance (Lemma~\ref{lem:extension_uhlmann_property_for_purefied_distance}), therefore for any $\varepsilon\ge 0$
        \begin{equation*}
            \hmincompc(B|C)_{\rho} \le  \hmincompc(AB|C)_{\rho} + \log(d_{A}) \;.\qedhere
        \end{equation*}
    \end{proof}

    \subsubsection{$\hunps[s]$ Equivalence}\label{app_subsec_cq-eqauive_proofs}
    \begin{lemma}[Equivalence for cq-States]\label{lem:cq-version1}
        For any cq-state $\rho_{XE}$ and $s\in\mathbb{N}$:
        \begin{equation*}
            \hmincompc[s](X|E)_{\rho} = \hunps[s](X|E)_{\rho} \;.
        \end{equation*}
    \end{lemma}
    \begin{proof}
        For this lemma, we use the same argument from~\cite{KRS09OperationalMeaningEntropy}. Let $\rho_{XE}= \sum_{x} p_x \proj{x}\otimes \rho^{x}_{E} $. Then, for any circuit, $\cE^s$, of size $s$ we have that
        \begin{align*}
            d_{X} F\left( (\mathbb{I}_X \otimes \cE^s)(\rho_{XE}), \proj{\Phi_{XX'}} \right) = \sum_{x} p_x \bra{x} \cE^s \left(\rho^{x}_{E} \right)\ket{x} \;.
        \end{align*}
        As such,
        \begin{equation}\label{eq:compminent-cq}
            \hmincompc[s](X|E)_{\rho} = -\log \left( \max_{\cE^s \in \cC(s)}\sum_{x} p_x \bra{x} \cE^s \left(\rho^{x}_{E} \right)\ket{x} \right) \;.
        \end{equation}
        Conversely, any guessing circuit, $\cF^{s}\in \cP(s)$, for the unpredictability entropy can be composed of a circuit of size $s$, which we also denote by $\cE^{s}\in \cC(s)$, and a subsequent measurement in the computational basis. Thus,
        \begin{align*}
            \hunps[s](X|E)_{\rho}
             & = -\log \left( \max_{\cF^s\in \cP(s)} \pr{\cF^s(\rho_{E}^x) = x}\right)                                         \\
             & =  -\log \left( \max_{\cE^s \in \cC(s)}\sum_{x} p_x \bra{x} \cE^s \left(\rho^{x}_{E} \right)\ket{x} \right) \;.
        \end{align*}
        This is the expression we derived for the computational min-entropy in the cq-case in Eq~\eqref{eq:compminent-cq} above.
    \end{proof}

    \begin{lemma}[Almost Equivalence for Smoothing cq-States]\label{lem:cq-version2}
        For any cq-state $\rho_{XE}$ and any $\varepsilon\ge 0 $, $s\in\mathbb{N}$, and $\ell = \log\dim(X) \in \mathbb{N}$:
        \begin{align*}
            \hmincompc(X|E)_{\rho} \geq \hunps[s]^{\varepsilon}(X|E)_{\rho} \geq \hmincompc[(s+l,\varepsilon)](X|E)_{\rho}  \;.
        \end{align*}
    \end{lemma}
    
    \begin{proof}
        The first inequality,
        \begin{align*}
            \hmincompc(X|E)_{\rho} & \geq \hunps[s]^{\varepsilon}(X|E)_{\rho} \; ,
        \end{align*}
        directly follows from  Lemma~\ref{lem:cq-version1} and the fact that any state
        $\tilde{\rho}_{XE} \in \cB^{\mathrm{cq}}_{\varepsilon}(\rho) $ also trivially satisfies $\tilde{\rho}_{XE} \in \cB_{\varepsilon}(\rho) $.

        To prove the latter inequality, let us consider any circuit of size $s$, $\cE^{(s)}$, which acts on a fully quantum state $\tilde{\rho}_{XE} \in \cB_{\varepsilon}\left(\rho\right)$. Moreover, let $\tilde{\rho}_{\mathcal{M}(X)E} = \sum_{x} \ketbra{x}{x} \otimes \bra{x} \tilde{\rho}_{XE}\ket{x}$  denote the post-measurement state after measuring the register $X$ in the computational basis. 
        Due to the data processing inequality for the purified distance~\cite[Proposition 3.2]{Tomamichel_2016} and the definition of~$\cB^{\mathrm{cq}}_{\varepsilon}(\rho)$, it must also hold that $\tilde{\rho}_{\mathcal{M}(X)E}\in \cB^{\mathrm{cq}}_{\varepsilon}(\rho)$.
        We now show that for every circuit, $\cE^{(s)}$, there exists another circuit, $\cE^{\prime \left(s+l\right)}$, of size at most $s+l$ such that
        \begin{align} \label{Eq: AFidIneq}
            F\left( \tilde{\rho}_{\mathcal{M}(X)\cE^s \left(E\right)}, \proj{\Phi_{XX'}} \right) \leq F\left( \tilde{\rho}_{X\cE^{\prime \left(s+l\right)} \left(E\right)}, \proj{\Phi_{XX'}} \right)
        \end{align}
        holds, where $\tilde{\rho}_{X\cE^{\prime \left(s+l\right)} (E)}= \cE^{\prime \left(s+l\right)} \left(\tilde{\rho}_{XE} \right)$ and $\tilde{\rho}_{\mathcal{M}(X)\cE^s (E)}= \cE^{(s)} \left(\tilde{\rho}_{\mathcal{M}(X)E} \right)$.
        This would then directly imply the second inequality. To see this, first note that
        \begin{align*}
            \hmincompc[(s+l,\varepsilon)](X|E)_{\rho} \leq
            -\log d_{X} \min_{\tilde{\rho}\in \cB_{\varepsilon}(\rho)}\max_{\cE^{\prime} }
            F\left( (\mathbb{I}_X \otimes \cE^{\prime})(\tilde{\rho}_{XE}), \proj{\Phi_{XX'}} \right) \; ,
        \end{align*}
        as the RHS is only optimizing over a subset of quantum circuits, $\cE^{\prime}$,  of size at most $s+l$, which are constructed from a circuit $\cE$ of size at most $s$. Moreover, given that one can construct a circuit $ \cE^{\prime \left(s+l\right)} $ from every circuit $\cE^{(s)}$ of size less than or equal to $s$, it follows from Eq.~\eqref{Eq: AFidIneq} that
        \begin{align*}
            \phantom{ \hmincompc[(s+l,\varepsilon)](X|E)_{\rho} \leq}
                  -\log d_{X} &\min_{\tilde{\rho}\in \cB_{\varepsilon}( \rho)}\max_{\cE^{\prime} }
            F\left( (\mathbb{I}_X \otimes \cE^{\prime})(\tilde{\rho}_{XE}), \proj{\Phi_{XX'}} \right) \\
            \leq & -\log d_{X} \min_{\tilde{\rho}\in \cB_{\varepsilon}(\rho)}\max_{\cE }
            F\left( (\mathcal{M} \otimes \cE)(\tilde{\rho}_{XE}), \proj{\Phi_{XX'}} \right)           \\
            =    & \hunps[s]^{\varepsilon}(X|E)_{\rho} \;.
        \end{align*}
        The last equality holds, as optimizing over all $\tilde{\rho}_{\mathcal{M}(X)E}$ such that $\tilde{\rho}_{XE} \in \cB_{\varepsilon}(\rho)$ is equivalent to optimizing over all $\rho^{\prime}_{XE}  \in \cB^{\mathrm{cq}}_{\varepsilon}(\rho)$; i.e. not only is $\tilde{\rho}_{\mathcal{M}(X)E}  \in \cB^{\mathrm{cq}}_{\varepsilon}(\rho)$, but
        for every ${\rho^{\prime}_{XE}  \in \cB^{\mathrm{cq}}_{\varepsilon}(\rho)}$, there exists a $\rho_{XE}  \in \cB_{\varepsilon}(\rho)$ such that $\tilde{\rho}_{\mathcal{M}(X)E} = \rho^\prime$. Combined, this would then yield the desired inequality,
        \begin{align*}
            \hmincompc[(s+l,\varepsilon)](X|E)_{\rho} \leq \hunps[s]^{\varepsilon}(X|E)_{\rho} \; .
        \end{align*}

        Let us now discuss how $\cE^{\prime \left(s+l\right)}$ is constructed. Without loss of generality, we assume that $\ell = \log\dim(X) = \log\dim(X^\prime)$. Moreover, let us denote the computational basis of $\mathcal{H}_{X}$ (and $\mathcal{H}_{X^\prime}$) by the set $\{ \ket{ \mathbf{k}} \}_{\mathbf{k} \in \{ 0,1 \}^{l}}$ and, correspondingly, let $\{ \ket{ \mathbf{k};\mathbf{l}}\}_{\mathbf{k}, \mathbf{l}\in \{ 0,1 \}^{\ell}}$ denote the computational basis of $\mathcal{H}_{X X^\prime}$, where $ \ket{ \mathbf{k};\mathbf{l}} \coloneq \ket{\mathbf{k}} \otimes \ket{\mathbf{l}}$. Then, in this basis, $\tilde{\rho}_{X\cE^s \left(E\right)}$ can be written as
        \begin{align*}
            \tilde{\rho}_{X\cE^s
                \left(E\right)} = \sum_{\mathbf{k},\mathbf{l},\mathbf{m},\mathbf{n}} \tilde{\rho}_{\mathbf{k},\mathbf{l},\mathbf{m},\mathbf{n}} \ketbra{\mathbf{k};\mathbf{l}}{\mathbf{m};\mathbf{n}} \; .
        \end{align*}
        For a given unitary circuits $Z^{\mathbf{i}} \coloneq Z^{i_1} \otimes \cdots \otimes Z^{i_\ell}$, where
        $Z= \begin{pmatrix}
                1 & 0  \\
                0 & -1
            \end{pmatrix}$
        and ${i_{1},\dots, i_{\ell} \in \{ 0,1 \}}$, one finds that
        \begin{align*}
            \tilde{\rho}_{Z^{\mathbf{i}}(X)\cE^s
            \left(E\right)} = \sum_{\mathbf{k},\mathbf{l},\mathbf{m},\mathbf{n}} (-1)^{\mathbf{i} \cdot \left( \mathbf{k} \oplus \mathbf{m}\right)}\tilde{\rho}_{\mathbf{k},\mathbf{l},\mathbf{m},\mathbf{n}} \ketbra{\mathbf{k};\mathbf{l}}{\mathbf{m};\mathbf{n}} \; .
        \end{align*}
        Note that whenever $\mathbf{k} = \mathbf{m}$ holds, $ (-1)^{\mathbf{i} \cdot \left( \mathbf{k} \oplus \mathbf{m}\right)}=1$. Conversely, for $\mathbf{k} \neq \mathbf{m}$, there exist $ 2^{n-1}$ choices of $\mathbf{i}$ such that $ (-1)^{\mathbf{i} \cdot \left( \mathbf{k} \oplus \mathbf{m}\right)}=1$ and $ 2^{n-1}$ choices such that $ (-1)^{\mathbf{i} \cdot \left( \mathbf{k} \oplus \mathbf{m}\right)}=-1$.\footnote{This is similar to the technique used in~\cite{Bennett_1996}.}

        Averaging over all unitary circuits $Z^{\mathbf{i}}$, the only terms that do not cancel out are those for which $\mathbf{k} = \mathbf{m}$, i.e.
        \begin{align*}
            F\left( \tilde{\rho}_{\mathcal{M}(X)\cE^s
                \left(E\right)}, \proj{\Phi_{XX'}} \right) &=  F\left( \frac{1}{2^\ell} \sum_{\mathbf{i}}\left( Z^{\mathbf{i}} \otimes \mathbb{I}_{X'} \right) \left(\tilde{\rho}_{X\cE^{\left(s\right)} \left(E\right)} \right), \proj{\Phi_{XX'}} \right)\\
                &= \frac{1}{2^\ell} \sum_{\mathbf{i}} F\left( \left( Z^{\mathbf{i}} \otimes \mathbb{I}_{X'} \right) \left(\tilde{\rho}_{X\cE^{\left(s\right)} \left(E\right)} \right), \proj{\Phi_{XX'}} \right) \; .
        \end{align*}
        Moreover, it holds that
        \begin{align*}
             & F\left( \tilde{\rho}_{\mathcal{M}(X)\cE^s
            \left(E\right)}, \proj{\Phi_{XX'}} \right)                                                    \\
             & = \frac{1}{2^\ell} \sum_{\mathbf{i}} F\left( \left( Z^{\mathbf{i}} \otimes \mathbb{I}_{X'} \right) \left(\tilde{\rho}_{X\cE^{\left(s\right)} \left(E\right)} \right), \proj{\Phi_{XX'}} \right)          \\
             & = \frac{1}{2^\ell} \sum_{\mathbf{i}} F\left( \tilde{\rho}_{X\cE^{\left(s\right)} \left(E\right)} , \left( Z^{\mathbf{i} \dagger} \otimes \mathbb{I}_{X'} \right)\left(\proj{\Phi_{XX'}} \right)\right)   \\
             & = \frac{1}{2^\ell} \sum_{\mathbf{i}} F\left( \tilde{\rho}_{X\cE^{\left(s\right)} \left(E\right)} , \left( \mathbb{I}_{X} \otimes  Z^{\mathbf{i} \star} \right)\left(\proj{\Phi_{XX'}} \right)\right)     \\
             & = \frac{1}{2^\ell} \sum_{\mathbf{i}} F\left( \left( \mathbb{I}_{X} \otimes  Z^{\mathbf{i} \intercal} \right)\left(\tilde{\rho}_{X\cE^{\left(s\right)} \left(E\right)}\right) , \proj{\Phi_{XX'}} \right) \\
             & = \frac{1}{2^\ell} \sum_{\mathbf{i}} F\left( \left( \mathbb{I}_{X} \otimes  Z^{\mathbf{i}} \right)\left(\tilde{\rho}_{X\cE^{\left(s\right)} \left(E\right)}\right) , \proj{\Phi_{XX'}} \right) \;.
        \end{align*}
        Here, $\star$ and $\intercal$ denote the complex conjugate and transpose, respectively. Moreover, for all $\mathbf{i}$, it holds that 
        \begin{align*}
            Z^{\mathbf{i} \star} = Z^{\mathbf{i} \intercal} = Z^{\mathbf{i}} \; .
        \end{align*}
    Since equality holds on average, there must thus exist at least one circuit, $Z^{\mathbf{i}}$, which can be appended to $\cE^{(s)}$, for which the desired inequality must hold. Moreover, this entire circuit has a circuit size of at most $s+l$.
    \end{proof}

    \subsubsection{$\mathrm{H}_{\min}^{\varepsilon}$ Relation}
    
    \begin{lemma}[Relation to Smooth Min-Entropy]
        Let $\rho_{AB}$ be a bipartite state, let $s\in\mathbb{N},\varepsilon\ge 0$
        \begin{align*}
            \hmincompc(A|B)_{\rho}                   & \ge \mathrm{H}_{\min}^{\varepsilon}(A|B)_{\rho}  \;, \\
            \lim_{s\to\infty} \hmincompc(A|B)_{\rho} & = \mathrm{H}_{\min}^{\varepsilon}(A|B)_{\rho}  \;.
        \end{align*}
    \end{lemma}

    \begin{proof}
        The proof directly results from the operational meaning of min-entropy as the negative log of the quantum correlation achievable by any quantum strategy, with no computational limitations~\cite[Theorem 2]{KRS09OperationalMeaningEntropy}, along with the fact that circuits of unbounded size using a universal gate set can approximate any unitary to arbitrary accuracy.
    \end{proof}

    \subsection{Proof of Lemma~\ref{lem:PropertiesCompMaxEnt}}
    \PropertiesCompMaxEnt*
    \subsubsection{Monotonicity}
    \begin{lemma}[Monotonicity in $s$]
        For any $s' \ge s$ and any bipartite state $\rho_{AB}$,
        \begin{equation*}
            \hmaxcompc(A|B)_{\rho} \le \hmaxcompc[\left(s',\varepsilon\right)](A|B)_{\rho}  \;.
        \end{equation*}
    \end{lemma}

    \begin{proof}
        Given any computational parameter, $s$, the computational max-entropy $\hmaxcompc(A|B)_{\rho}$ is defined via maximization over all quantum circuits $\cE_{C \to A'}$ of size at most $s$. If we increase $s$ to $s'\ge s$, the optimization would be over all circuits of size $s'$, which includes all circuits of size $s$. Therefore, the maximal fidelity achievable can only grow and thus:
        \begin{equation*}
            \hmaxcompc(A|B)_{\rho} \le \hmaxcompc[\left(s',\varepsilon\right)](A|B)_{\rho} \;. \qedhere
        \end{equation*}
    \end{proof}

    \begin{lemma}[Monotonicity in $\varepsilon$]
        For any $\varepsilon' \ge \varepsilon$ and any bipartite state $\rho_{AB}$,
        \begin{equation*}
        \hmaxcompc[\left(s,\varepsilon'\right)](A|B)_{\rho}
            \le \hmaxcompc(A|B)_{\rho}   \;.
        \end{equation*}
    \end{lemma}

    \begin{proof}
        Given any smoothing parameter, $\varepsilon$, the smooth computational max-entropy $\hmaxcompc(A|B)_{\rho}$ is defined by a purified ball around $\rho_{ABC}$.
        If we increase $\varepsilon$ to $\varepsilon' \ge \varepsilon$, the optimization would be over a larger ball centered around the same state. Thus, all the same circuits and states are available for the larger $\varepsilon'$, along with additional states. Therefore, only a smaller minimum can be achieved by $\varepsilon'$-smoothing, meaning
        \begin{equation*}
             \hmaxcompc[\left(s,\varepsilon'\right)](A|B)_{\rho} \le \hmaxcompc(A|B)_{\rho} \;. \qedhere
        \end{equation*}
    \end{proof}
\subsubsection{$\mathrm{H}_{\max}^{\varepsilon}$ Relation}
    \begin{lemma}[$\mathrm{H}_{\max}^{\varepsilon}$ Relation]
        $\mathrm{H}_{\max}^{\varepsilon}$ Relation:  For any $\rho_{AB}$, $s \in \mathbb{N}$, and $\varepsilon \geq 0$,
        \begin{align}
            \hmaxcompc(A|B)_{\rho}                   & \le \mathrm{H}_{\max}^{\varepsilon}(A|B)_{\rho}\;, \\
            \lim_{s\to\infty} \hmaxcompc(A|B)_{\rho} & = \mathrm{H}_{\max}^{\varepsilon}(A|B)_{\rho}  \;.
        \end{align}
    \end{lemma}

    \begin{proof}
Let $\tilde{\rho}_{AB} \in \cB_{\varepsilon}(\rho_{AB})$ denote the optimizer for the smooth max-entropy, i.e.
\begin{align*}
     \mathrm{H}_{\max}(A|B)_{\tilde{\rho}} =  \mathrm{H}_{\max}^{\varepsilon}(A|B)_{\rho} \; .
\end{align*}
Moreover, let $\ket{\tilde{\rho}^{\textrm{pg}}}_{ABC} = \diff{\sqrt{d_{AB}}} 
        \left(\sqrt{\tilde{\rho}_{AB}} \otimes \mathbb{I}_C\right)
        \ket{\Phi}_{AB|C}$ be the pretty good purification for this optimizer. It then follows from Lemma~\ref{lem:duality_for_smooth_info_min_max_entropy}, Lemma~\ref{Lem:PropertiesCompMinEnt}, and the definition of the computational max-entropy that
        \begin{align*}
            \mathrm{H}_{\max}^{\varepsilon}(A|B)_{\rho}  &= \mathrm{H}_{\max}(A|B)_{\tilde{\rho}}\\
            &= -\mathrm{H}_{\min}(A|C)_{\tilde{\rho}^{\textrm{pg}}} \\
            &\geq -\hmincompc[s](A|C)_{\tilde{\rho}^{\textrm{pg}}} \\
            &= \diff{ \hmaxcompc[s](A|B)_{\tilde{\rho}}} \\
            & \geq \hmaxcompc(A|B)_{\rho}
            \; .
        \end{align*}
This proves the first inequality,
        \begin{equation*}
            \hmaxcompc(A|B)_{\rho} \le \mathrm{H}_{\max}^{\varepsilon}(A|B)_{\rho}\;.
        \end{equation*}
To prove the second claim, recall that the information-theoretic max entropy is given by
        \begin{equation*}
            \mathrm{H}_{\max}(A|B)_{\rho} = \max_{\sigma_{B}\in \cS_{\circ}(B)}\log F(\diff{\rho_{AB}},\1_{A}\otimes\sigma_{B})\;.
        \end{equation*}
By the extension property of the fidelity, for any extension of $\mathbb{I}_{A}\otimes\sigma_{B}$ there is an extension of $\rho_{AB}$ such that the fidelity is preserved. In particular, for the extension $\proj{\Psi_{AA'}}\otimes\sigma_{B}$ there exists a state $\rho_{ABA'}$ such that:
        \begin{equation*}
            \max_{\sigma_{B}\in \cS_{\circ}(B)}\log F(\rho_{AB},\1_{A}\otimes\sigma_{B})= \max_{\sigma_{B}\in \cS_{\circ}(B)}\log F(\rho_{ABA'},\proj{\Psi_{AA'}}\otimes\sigma_{B})\;.
        \end{equation*}
        All purifications of a given state are equivalent up to isometry on the purifying system, therefore, for any purification $\rho_{ABC}$, there is a channel $V_{C\to A'}$ such that
        \begin{equation*}
            (\mathbb{I}_{AB}\otimes V_{C\to A'})\rho_{ABC} = \rho_{ABA'}\;.
        \end{equation*}
        Since the fidelity is defined by maximizing over all purifications, the state $\rho_{ABA'}$ is the extension with the highest fidelity possible with $\proj{\Psi_{AA'}}\otimes\sigma_{B}$.
        On the other hand our computational max entropy is defined by a specific purification $\rho_{ABC}$ and a limited set of channels, implementable by circuits of size at most $s$, $(\mathbb{I}_{AB}\otimes \cE^{s}_{C\to A'})\rho_{ABC}$. From the universality of the gate set, at the limit $s\to\infty$, the channel $V_{C\to A'}$ can be approximated to arbitrary precision, and therefore
        \begin{equation*}
            \lim_{s\to\infty} \hmaxcompc(A|B)_{\rho} = \mathrm{H}_{\max}^{\varepsilon}(A|B)_{\rho} \;. \qedhere
        \end{equation*}
    \end{proof}

\InfoCompMinSep*
    
\begin{proof}
    Fix $s\in\mathbb{N}$. By~\cref{lem:counting_hard_invers_fixed_s}, choose $n$ and a permutation $f:\{0,1\}^n\to\{0,1\}^n$ such that every size-$s$ circuit inverts $y=f(x)$ with success at most $2^{-n/2}$. Define the fully classical state
    \begin{equation*}
        \rho_{AB}=2^{-n}\sum_{x\in\{0,1\}^n}\ketbra{x}{x}_A\otimes\ketbra{f(x)}{f(x)}_B \;.
    \end{equation*}
    Since $f$ is a permutation, $A=f^{-1}(B)$ deterministically and thus $\mathrm{H}_{\min}(A|B)_\rho=0$.
    Moreover, $\rho_{AB}$ is a fully classical, in particular it is classical quantum so by the equivalence for cq-states~\cref {lem:cq-version1} we see:
    \begin{equation*}
        \hmincompc[s](A|B)_\rho
        =\hunps[s](A|B)_\rho
        =-\log p_{\mathrm{guess}}^s(A|B)
        \ge -\log 2^{-n/2}=n/2 \;. \qedhere
    \end{equation*}
\end{proof}

    \section{Hard to Invert Permutations}\label{Appendix_hard_permutations_counting_compleatness_proofs}

    \HardPermutations*

    \begin{proof}
        Using the probabilistic method, we use the fact that there are many more permutations on $n$ bits than there are circuits using $s$ gates.
        For any fixed $s\in \mathbb{N}$, we can bound the number of circuits using $s$ gates using a counting argument. Note that $s$ gates, which each act on up to $m$ qubits, can act on at most $s \cdot m$ qubits. It thus suffices to assume there are at most $s \cdot m$ ancilla qubits. There are $|G|$ choices of gates and at most $ (n+ms)^m$ choices of qubits for each gate in $G$ to act on.
        The number, $N_{n,s}$, of distinct circuits on $n$ qubits using at most $s$ gates out of the universal gate set $G$ is therefore upper bounded by: 
        \begin{equation*}
            N_{n,s} \leq \left(|G| (n+ms)^m\right)^s \;.
        \end{equation*}
        For any fixed circuit $C$ and $y\in \bits^n$, for a random permutation $f$:
        \begin{equation*}
            \Pr_{f}\left[C(y) = f^{-1}(y)\right] = 2^{-n} \;.
        \end{equation*}
        The expectation value for a random $y\in\bits^n$ is thus
        \begin{equation*}
            \mathbb{E}_{f}\left[\Pr_{y}{C(y) = f^{-1}(y)}\right] = 2^{-n} \;.
        \end{equation*}
        Denoting $p_{C}(f) = \Pr_{y}\left[C(y) = f^{-1}(y)\right]$, it follows by Markov's inequality that
        \begin{equation*}
            \Pr_{f}\left[p_{C}(f)\ge 2^{-n/2}\right] \le 2^{-n/2} \;.
        \end{equation*}
        By a union bound argument over all $s$-sized circuits
        \begin{equation*}
            \Pr_{f}\left[\exists C \in \cC(s): p_{C}(f)> 2^{-n/2}\right] \le \left(|G| (n+ms)^m\right)^s 2^{-n/2} \;.
        \end{equation*}
       Since $\left(|G| (n+ms)^m\right)^s$ scales polynomially in $n$, there must exist an $n^\prime$ such that for all $n \geq n^\prime$
        \begin{align*}
            \left(|G| (n+ms)^m\right)^s < 2^{n/2} \; .
        \end{align*}
        Therefore, for such $n$, it must hold that
         \begin{equation*}
            \Pr_{f}\left[\exists C \in \cC(s): p_{C}(f)> 2^{-n/2}\right] <1 \;,
        \end{equation*}
        and there exists a permutation on $n$ bits that no $s$-sized circuit can invert with probability higher than $2^{-n/2}$.
    \end{proof}

\end{document}